\documentclass[letterpaper, 10 pt,journal]{IEEEtran}

\usepackage{graphicx}
\usepackage{amsmath,amssymb}   
\usepackage{tikz}
\usepackage{cite}
\usepackage{subfigure}

\usepackage{amsthm}
\theoremstyle{remark}
\newtheorem{thm}{Theorem}
\newtheorem{defn}[thm]{Definition}
\newtheorem{assu}[thm]{Assumption}
\newtheorem{algo}[thm]{Algorithm}
\newtheorem{lem}[thm]{Lemma}
\newtheorem{lrop}[thm]{Proposition}
\newtheorem{coro}[thm]{Corollary}
\theoremstyle{remark}
\newtheorem{rmk}[thm]{Remark}

\usepackage{mathtools}

\title{Data-driven inference on optimal input-output properties of polynomial systems with focus on nonlinearity measures }
\author{Tim Martin and Frank Allg{\"o}wer
	\thanks{{We thank the foundation by the Deutsche Forschungsgemeinschaft (DFG, German Research Foundation) under Germany's Excellence Strategy - EXC 2075 - 390740016 and under grant 468094890. We acknowledge the support by the Stuttgart Center for Simulation Science (SimTech).}}
	\thanks{The authors are with the University of Stuttgart, Institute for Systems Theory and Automatic Control, Germany. (e-mail: tim.martin@ist.uni-stuttgart.de; frank.allgöwer@ist.uni-stuttgart.de).}
}

\begin{document}

\IEEEoverridecommandlockouts

\IEEEpubid{\begin{minipage}{\textwidth}\ \\[40pt] \copyright 2022 IEEE. Personal use of this material is permitted. Permission from IEEE must be obtained for all other uses, in any current or future media,
		including reprinting/republishing this material for advertising or promotional purposes, creating new collective works, for resale or redistribution to servers
		or lists, or reuse of any copyrighted component of this work in other works.\end{minipage}}

\maketitle
\pagestyle{empty}	


\begin{abstract}
	
In the context of dynamical systems, nonlinearity measures quantify the strength of nonlinearity by means of the distance of their input-output behaviour to a set of linear input-output mappings. In this paper, we establish a framework to determine nonlinearity measures and other optimal input-output properties for nonlinear polynomial systems without explicitly identifying a model but from a finite number of input-state measurements which are subject to noise. To this end, we deduce from data for the unidentified ground-truth system three possible set-membership representations, compare their accuracy, and prove that they are asymptotically consistent with respect to the {amount} of samples. Moreover, we leverage these representations to compute guaranteed upper bounds on nonlinearity measures and the corresponding optimal linear approximation model via semi-definite programming. Furthermore, we {extend} the {established} framework to determine optimal input-output properties described by time domain hard integral quadratic constraints.
\end{abstract}

\begin{IEEEkeywords}
Data-driven system analysis, identification for control, polynomial dynamical systems.
\end{IEEEkeywords}

\section{Introduction}\label{Intrduction}

Most controller design techniques for nonlinear systems require a precise model of the system. However, the concurrently increasing complexity of plants in engineering leads to time-consuming modelling by first principles. Therefore, data-driven controller design techniques have been developed where a controller is derived from measured trajectories of the plant.\\\indent
For that purpose, a two-step procedure is usually applied where first a model of the control plant is retrieved by system identification techniques in order to apply controller design methods afterwards. To recover closed-loop stability from controller design techniques with inherent closed-loop guarantees, an estimation of the model error is required which is an active research field even for linear time-invariant (LTI) systems \cite{SysIdLin}. On the other hand, recent interests consider a controller design directly from measured trajectories with rigorous closed-loop guarantees. In this context, data-driven approaches for nonlinear systems include virtual reference feedback tuning \cite{VRT}, adaptive control \cite{Adaptive}, and set-membership \cite{MilanseControlDesign} and \cite{DePersisSOS}. \cite{DataSurvey} provides a broader overview of such kind of methods.\\\indent
In this paper, we follow the alternative direction of \cite{MontenbruckLipschitz} where system theoretic properties, as dissipativity \cite{DissiWillems}, of an unknown system are determined from data. System theoretic properties have a large relevance in system analysis and robust controller design as they provide insights into the system and facilitate a controller design without knowledge of the system dynamics. Thus, we can leverage the determination of these properties from measured trajectories for a data-driven controller design. Further motivations for deriving a controller by means of system properties are a modular controller design for large-scale systems, well-established feedback theorems \cite{Khalil} but also recent control methods, e.g., for network control systems \cite{Network}, and uncertainty characterization in application fields as (soft) robotics \cite{Softr}. \\\indent
If the system property, that represents the system behaviour, is chosen inappropriately, then this design ansatz can lead to conservative control performance. Therefore, we consider the extensive framework of integral quadratic constraints (IQCs) which achieve, compared to dissipativity, a more informative description of input-output properties, and hence a less conservative robust controller design \cite{IQCSyn}. A certain class of IQCs are nonlinearity measures (NLMs) \cite{FA_NLM} where the strength of nonlinearity of a dynamical system is quantified by means of an \textquoteleft optimal\textquoteright\ linear approximation model of the nonlinear input-output behaviour.\\\indent 
The estimation of system properties of this nature has been examined for a long time but to mention recent research, \cite{OneShot} and \cite{AnneIQC} determine dissipativity and IQCs, respectively, over a finite time horizon from a noise-free input-output trajectory for LTI systems. For any finite time horizon, \cite{AnneDissi} guarantees dissipativity properties from noisy input-state samples using the data-based system representation from \cite{Groningen}. Recently, we established in \cite{MartinDissi} two data-based set-membership frameworks 
to verify dissipativity properties for unidentified polynomial systems via sum-of-squares (SOS) optimization. Note that the basis of \cite{AnneDissi}, \cite{MartinDissi}, and this work is a set-membership ansatz with deterministic noise description as it allows the direct application of robust control techniques to determine system properties by semi-definite programming (SDP). While a comparable parameter description by probability distributions from Bayesian estimation would yield to non-tractable optimization for verifying system properties, a probabilistic noise description could lead to less conservative parameter estimations.\\\indent
The contributions of this work are the following. First, while the determination of optimal IQCs, contrary to dissipativity \cite{MartinDissi}, requires intrinsically the non-convex optimization over a linear filter, we provide a comprising framework to retrieve for an unidentified polynomial system from noisy data the optimal, i.e., the tightest, system property specified by a certain class of IQC using computation tractable linear matrix inequalities (LMIs) including SOS multipliers. In particular, Algorithm~\ref{Algorithm1} yields together with Theorem~\ref{SynAENLM} to a data-driven inference on NLMs and infers together with Corollary~\ref{CoroSynIQC} on more general IQCs. In contrast to \cite{IterativeNLM} and \cite{MartinNLM}, we conclude on NLMs for any finite time horizon and for optimized linear approximation models which are specified by a general linear state-space model, and hence allow for less conservative inferences. Second, since we demonstrated in \cite{MartinDissi} that a cumulative noise description, as suggested in \cite{Groningen}, might lead to conservative parameter estimations if the noise is actually bounded in each time step, we constitute in Proposition~\ref{Primal_elli} and \ref{WSigma} two supersets for the unidentified coefficients which are more accurate than the superset in \cite{Groningen}. Third, the investigation of the asymptotic accuracy of these supersets in Theorem~\ref{ConverW}, \ref{Converp}, and \ref{SetofunFaCoeAv} is intrinsically an extension of \cite{SetMem_Con} and \cite{SetMem_MPC}, but shows for the first time the strong connection of the recent data-driven approaches, e.g., \cite{DePersisSOS, AnneDissi}, and \cite{Groningen}, and the set-membership literature, e.g., \cite{SetMem_Con} and \cite{Milanese}. Hence, this link could lead to further insights into the recent data-driven methods in the future.\\\indent
The paper is organized as follows. In Section~\ref{PreSet}, we state the problem setup and the connection of NLMs to system properties from the control literature. Section~\ref{Supersets} provides a data-based characterization of the unidentified coefficients which is the basis to establish the framework for determining NLMs from data in Section~\ref{SecGeneralSOS}. Section~\ref{Ext} contains extensions of this framework to determine, e.g., tight IQCs of certain classes. Section~\ref{DataPre2} compares the set-membership characterization from Section~\ref{Supersets} to two others and investigates their accuracy and asymptotic consistency.  
\IEEEpubidadjcol

\section{Problem formulation}\label{PreSet}

\subsection{Notation}

Along the paper, let $\mathbb{N}_0$ denote the set of all natural numbers including zero, $\mathbb{N}_{[a,b]}=\{n\in\mathbb{N}_0:a\leq n\leq b\}$ the set of all integers in the interval $[a,b]$, and $\mathbb{N}_{\geq0}=\mathbb{N}_{[0,\infty)}$. Analogous definitions hold for the set of real numbers $\mathbb{R}$. The floor value of a scalar $s$ is denoted by $\lfloor s\rfloor$. Let $\partial M$ denote the boundary of a set $M$ and $\oplus$ denote the Minkowski addition of two sets. Furthermore, the probability of an event $E$ is denoted by $\text{Pr}({E})$. The Euclidean norm of a vector $v\in\mathbb{R}^n$ is denoted as $||v||_2$. $I_n$ denotes the $n\times n$ identity matrix and $0$ the zero matrix of suitable dimensions. Moreover, we write $\text{vec}(A)\in\mathbb{R}^{nm}$ for the vectorization of $A\in\mathbb{R}^{n\times m}$ by stacking its columns. For some matrices $A_1$, $A_2$, and $A_3$, we write the block diagonal matrix
\begin{equation*}
	\text{diag}(A_1,A_2\big| A_3)=\begin{bmatrix}\begin{array}{c|c}
	\begin{matrix}A_1 & 0\\0 & A_2\end{matrix} & 0 \\\hline 0 & A_3 
	\end{array}	\end{bmatrix}{.}
\end{equation*}
Furthermore, we introduce for matrices $M,N\in\mathbb{R}^{m\times n}$ the inner product $\left<M,N\right>_\text{Fr}=\text{tr}(M^TN)$ which implies the Frobenius norm $||M||_\text{Fr}=\sqrt{\left<M,M\right>_\text{Fr}}$.\\\indent
Let $\ell_2^p$ denote the vector space of infinite sequences of real numbers $u:\mathbb{N}_0\rightarrow\mathbb{R}^p$ for which $||u||_{\ell_2}=(\sum_{t=0}^{\infty}||u(t)||_2^2)^{1/2}<\infty$. 
By convention, let $\ell_{2e}^p$ be the space of infinite sequences satisfying $u_T\in\ell_2^{p}$ for all $T\in\mathbb{N}_0$ where $(\cdot)_T$ denotes the truncation operator
\begin{equation*}
	u_T(t)=\left\{
	\begin{array}{ll}
	u(t) & \text{for}\ t\leq T \\
	0 &  \text{for}\ t>T \\
	\end{array}
	\right. .
\end{equation*}\indent
For the investigation of polynomial systems, we define $\mathbb{R}[x]$ as the set of all polynomials $p$ in $x=\begin{bmatrix}
x_1 & \cdots & x_n\end{bmatrix}^T\in\mathbb{R}^n$, i.e.,
\begin{equation*}
	p(x)=\sum_{\alpha\in\mathbb{N}_0^n,|\alpha|\leq d} a_\alpha x^\alpha,
\end{equation*}
with vectorial indices $\alpha={\begin{bmatrix}\alpha_1 & \cdots & \alpha_n\end{bmatrix}^T}\in\mathbb{N}_0^n$, $|\alpha|=\alpha_1+\cdots+ \alpha_n$, monomials $x^\alpha=x_1^{\alpha_1}\cdots x_n^{\alpha_n}$, real coefficients $a_\alpha\in\mathbb{R}$, and $d$ as the degree of $p$. In addition, we define the set of all $m$-dimensional polynomial vectors $\mathbb{R}[x]^m$ and $m\times n$ polynomial matrices $\mathbb{R}[x]^{m\times n}$ where each entry is an element of $\mathbb{R}[x]$. The degree of a polynomial vector or matrix corresponds to the largest degree of its elements. For a polynomial matrix $P\in\mathbb{R}[x]^{n\times n}$ with even degree, if there exists a matrix $Q\in\mathbb{R}[x]^{m\times n}$ such that $P=Q^TQ$, then $P$ is an SOS matrix or SOS polynomial for $n=1$. {$\text{SOS}[x]^{n\times n}$ denotes the set of all ${n}\times {n}$ SOS matrices and $\text{SOS}[x]$ the set of all SOS polynomials.}

\subsection{NLM as input-output property}\label{NLMIntro}

In this section, we introduce a measure of the nonlinearity of the input-output behaviour of dynamical systems. We also relate this measure to various system properties from the control literature.\\\indent
As common in nonlinear control \cite{Khalil}, the input-output behaviour of dynamical systems can be represented by an operator $H:\mathcal{U}\subseteq\mathcal{\ell}_{2e}^{n_u}\rightarrow\mathcal{Y}\subseteq\mathcal{\ell}_{2e}^{n_y}$ that maps each input signal uniquely to an output signal and that satisfies $H(u_T)_T=H(u)_T$ for all $T\in\mathbb{N}_0$ to ensure causality. In the following, we assume that $H$ is stable, i.e., its $\ell_2$-gain
\begin{equation}\label{Def_ell2_gain}
	||H||^{\mathcal{U}}=\sup_{u\in\mathcal{U}\backslash \{0\},T\in\mathbb{N}_0}\frac{||{H}(u)_T||_{\ell_2}}{||u_T||_{\ell_2}}
\end{equation}
is finite. To quantify the nonlinearity of such an operator, \cite{FA_NLM} suggests the following notion of NLM.
\begin{defn}[Additive error NLM (AE-NLM)]\label{NLMDef}
The nonlinearity of a causal and stable nonlinear system $H:\mathcal{U}\rightarrow\mathcal{Y}$ is measured by
\begin{align}\label{AE-NLM}
	\Phi^{\mathcal{U},\mathcal{G}}_{\text{AE}}=\inf_{G\in\mathcal{G}}\sup_{\substack{u\in\mathcal{U}\backslash \{0\},T\in\mathbb{N}_0}}\frac{||H(u)_T-G(u)_T||_{\ell_2}}{||u_T||_{\ell_2}},
\end{align}
where $\mathcal{G}$ is a set of stable and linear mappings $G:\mathcal{U}\rightarrow\mathcal{\ell}_{2e}^{n_y}$.
\end{defn}

By the stability of $H$ and $G$, the AE-NLM exists as clarified in \cite{TS}. While the supremum of \eqref{AE-NLM} corresponds to the $\ell_2$-gain from input $u$ to the error $e(u)=H(u)-G(u)$ as illustrated in Figure~\ref{Fig:NLMDef}, the infimum yields the linear system $G^*$ in $\mathcal{G}$ that minimizes the $\ell_2$-gain of the error model $\Delta=H-G$. 
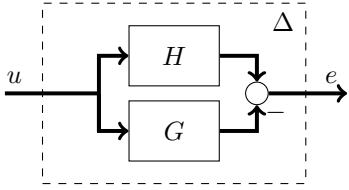
\begin{figure}
	\begin{center}
		\begin{tikzpicture}[scale=0.5]

\path (2,0) node[rectangle,draw,minimum width=12mm, minimum height=8mm](G) {$G$};
\path (G)+(0,2) node[rectangle,draw,minimum width=12mm, minimum height=8mm](Delta) {$H$};
\path (G)+(0,1) node[rectangle,draw,minimum width=35mm, minimum height=24mm, dashed](N) {};
\path (G)+(2.9,2.9) node(N1) {$\Delta$};

\path (G)+(2.2,1) node[circle,draw,inner sep=0pt, minimum size=3mm](rSum){};
\path (rSum)+(0.5,-0.5) node(minus){$-$};
\draw[-,line width=1.5pt] (G)+(-4.5,1) -- (0,1) node[pos=0.1, above=3pt,inner sep=0pt](HP2){$u$};
\draw[->,line width=1.5pt] (G)+(-2,0) |- (Delta);
\draw[->,line width=1.5pt] (G)+(-2,1) |- (G);
\draw[->,line width=1.5pt] (Delta) -| (rSum);
\draw[->,line width=1.5pt] (G) -| (rSum);
\draw[->,line width=1.5pt] (rSum) -- +(2.4,0) node[pos=0.8, above=3pt,inner sep=0pt](HP2){$e$};
    	
\end{tikzpicture}
	\end{center}
	\caption{Illustration of the AE-NLM.}
	\label{Fig:NLMDef}
\end{figure}
Therefore, $G^*$ can be seen as the \textquoteleft optimal\textquoteright\ linear approximation of the nonlinear system behaviour of $H$ and could be  exploited as linear surrogate model of $H$ with known error bound, e.g., for a robust controller design with rigorous closed-loop guarantees. 
Since the error of a global linear approximation for a general nonlinear system is mostly unbounded, we define the AE-NLM locally over $\mathcal{U}\subseteq\mathcal{\ell}_{2e}^{n_u}$. Furthermore, \cite{TS} shows that the AE-NLM is equal to $||H||^{\mathcal{U}}$ with $G^*=0$ if $H$ is strongly nonlinear and the NLM is zero if $H$ has a linear input-output behaviour.\\\indent
In the sequel, we relate the AE-NLM to other system properties from control theory. First, Definition~\ref{NLMDef} includes the conic relations from \cite{Zames} as special case for a static center $\mathcal{G}=\{G=c\,\text{id}: c\in\mathbb{R}\}$ with $\text{id}:u\mapsto u$. Since $\Phi^{\mathcal{U},\mathcal{G}}_{\text{AE}}$ can be seen as the width of the tightest cone with center $G^*$ and containing $H$, the width for a static center is larger than for a dynamic center. Thus, we conclude that a stabilizing controller, obtained from a dynamic center, can be confined in a larger cone, and hence is less conservative than a controller by applying the feedback theorem from \cite{Zames}. Furthermore, we also showed in \cite{MartinNLM} that AE-NLM can be described as dynamic conic sector \cite{Teel} from which a feedback theorem can be deduced via topological graph separation.\\\indent
Second, if the nonlinear input-output mapping $H$ is specified by a nonlinear state-space representation
\begin{equation}\label{NLsys}
	H:\left\{\begin{aligned} x(t+1)&=f(x(t),u(t)), x(0)=0 \\ y(t)&=h(x(t),u(t)),t\in\mathbb{N}_0 \end{aligned}\right.
\end{equation}
with input $u(t)\in\mathbb{U}\subseteq\mathbb{R}^{n_u}$, state $x(t)\in\mathbb{X}\subseteq\mathbb{R}^{n_x}$, and output $y(t)\in\mathbb{Y}\subseteq\mathbb{R}^{n_y}$, then dissipativity theory \cite{DissiWillems} constitutes an elaborate framework to characterize input-output properties by simple inequality conditions. Contrary to \cite{DissiWillems}, we give here a local notion of dissipativity.
\begin{defn}[Dissipativity]\label{DissiDef}
	System~\eqref{NLsys} is dissipative on $\mathbb{Z}=\mathbb{X}\times\mathbb{U}$ {regarding} the supply rate $s:\mathbb{Z}\rightarrow\mathbb{R}$ if there exists a continuous storage function $\lambda:\mathbb{X}\rightarrow\mathbb{R}_{\geq0}$ such that
	\begin{align}\label{dissipativityInqu}
	\lambda(f(x,u))-\lambda(x)\leq s(x,u),\quad \forall (x,u)\in\mathbb{Z}.
	\end{align}
	\vspace{-0.3cm}
\end{defn} 

In particular, we are interested in the supply rate
\begin{equation}\label{SupplyRate2}
	s(x,u)=\gamma ||u||_2^2-\frac{1}{\gamma}||y||_2^2,
\end{equation}
with $y=h(x,u)$, because the corresponding dissipativity property is connected to gains of systems within invariant sets as follows from \cite{Gain} (Proposition 3.1.7)
\begin{lrop}[Gains of systems]\label{GainSys}
For an operator \eqref{NLsys}, assume $\mathbb{X}$ is invariant under $x(0)$ and $u(t)\in\mathbb{U}$. Then, its $\ell_2$-gain \eqref{Def_ell2_gain} with $\mathcal{U}=\{u\in\ell_2^{n_u}: u(t)\in\mathbb{U}, \forall t\in\mathbb{N}_0\}$ is given by the smallest $\gamma\geq0$ such that \eqref{NLsys} is dissipative on $\mathbb{Z}=\mathbb{X}\times\mathbb{U}$ regarding the supply rate $s(x,u)=\gamma ||u||_2^2-\frac{1}{\gamma}||y||_2^2$ and admits a storage function with $\lambda(0)=0$.
\end{lrop}

In Section~\ref{SecGeneralSOS}, this connecting of dissipativity and system gains plays a crucial role as the $\ell_2$-gain of the error system $\Delta$ is equal to the AE-NLM. Note that Proposition~\ref{GainSys} requires a state-space instead of an input-output representation of the system in order to provide system properties over arbitrary time horizons.\\\indent 
We already mentioned that AE-NLM generalizes the conic relations from \cite{Zames} by a dynamic center. From another viewpoint, NLMs constitute a special case of IQCs. Although our focus lies on NLMs, IQCs build an attractive and frequently-studied framework to describe and work with a large class of input-output properties, e.g., compare \cite{IQC_Scherer}. Therefore, we will adapt our main result for deriving AE-NLM to also determine tight IQCs of certain classes in Section~\ref{Ext}. {While \cite{IQC_Rantzer} originally introduces IQCs in the frequency domain, we consider here only time domain IQCs and refer to \cite{IQC_intro} for a more detailed introduction of IQCs.}
\begin{defn}[Time domain hard IQC]\label{IQC}
System $H:u\in\ell_{2e}^{n_u}\mapsto y\in\ell_{2e}^{n_y}$ satisfies the time domain hard IQC with matrix $M\in\mathbb{R}^{n_r\times n_r}$ and stable LTI system
\begin{equation}\label{IQCFilter}
	\Psi:\left\{\begin{aligned} x_\Psi(t+1)&=A_\Psi x_\Psi(t)+B_{\Psi u} u(t)+B_{\Psi y} y(t)\\
	x_\Psi(0)&=0 \\ 
	r(t)&=C_\Psi x_\Psi(t)+D_{\Psi u} u(t)+D_{\Psi y} y(t) \end{aligned}\right. 
\end{equation}
if, for all $N\in\mathbb{N}_0$ and $r(t)\in\mathbb{R}^{n_r}$ given by \eqref{IQCFilter}, it holds
\begin{align}\label{IQCIneq}
	\sum_{t=0}^{N}r(t)^TMr(t)\geq 0.
\end{align}
\vspace{-0.1cm}
\end{defn}  

Definition~\ref{IQC} can be illustrated as in Figure~\ref{Fig.IQC}, i.e., signal $r(t)$ corresponds to the filtered input and output of system $H$ by $\Psi$. 
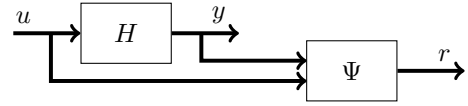
\begin{figure}
	\begin{center}
		\begin{tikzpicture}[scale=0.5]

\path (2,0) node[rectangle,draw,minimum width=12mm, minimum height=8mm](G) {$H$};
\path (G)+(6,-1) node[rectangle,draw,minimum width=12mm, minimum height=8mm](Delta) {$\Psi$};

\path (Delta)+(-0.95,-0.27) node(HP1) {};
\path (Delta)+(-0.95,0.27) node(HP2) {};

\draw[->,line width=1.5pt] (G)+(-3,0) -- (G) node[pos=0.15, above=3pt,inner sep=0pt](HP){$u$};
\draw[->,line width=1.5pt] (Delta) -- +(3,0) node[pos=0.7, above=3pt,inner sep=0pt](HP){$r$};
\draw[->,line width=1.5pt] (G)+(-2,0) |- (HP1);
\draw[->,line width=1.5pt] (G)+(2,0) |- (HP2);
\draw[->,line width=1.5pt] (G) -- +(3,0) node[pos=0.7, above=3pt,inner sep=0pt](HP){$y$};

\end{tikzpicture}
	\end{center}
	\caption{Graphical illustration of an IQC.}
	\label{Fig.IQC}
\end{figure}
The time domain IQC \eqref{IQCIneq} corresponds to a sum quadratic constraint on the filter output $r$. 
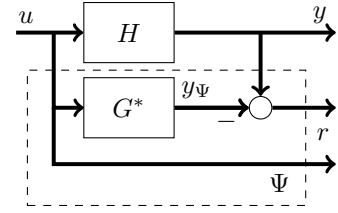
\begin{figure}
	\begin{center}
		\begin{tikzpicture}[scale=0.5]

\path (2,0) node[rectangle,draw,minimum width=12mm, minimum height=8mm](G) {$H$};
\path (G)+(0,-2) node[rectangle,draw,minimum width=12mm, minimum height=8mm](Delta) {$G^*$};
\path (Delta)+(3.5,0) node[circle,draw,inner sep=0pt, minimum size=3mm](rSum){};

\path (Delta)+(-0.95,0) node(HP1) {};

\draw[->,line width=1.5pt] (G)+(-3,0) -- (G) node[pos=0.15, above=3pt,inner sep=0pt](HP){$u$};
\draw[->,line width=1.5pt] (Delta) -- (rSum) node[pos=0.7, below=1pt,inner sep=0pt](HP){$-$};
\draw[->,line width=1.5pt] (Delta) -- (rSum) node[pos=0.3, above=3pt,inner sep=0pt](HP){$y_\Psi$};
\draw[->,line width=1.5pt] (G)+(-2,0) |- (HP1);
\draw[->,line width=1.5pt] (G)+(3.5,0) -- (rSum);
\draw[->,line width=1.5pt] (G) -- +(5.5,0) node[pos=0.9, above=3pt,inner sep=0pt](HP){$y$};
\draw[->,line width=1.5pt] (rSum) --+(2,0) node[pos=0.8, below=7pt,inner sep=0pt](HP){$r$};
\draw[->,line width=1.5pt] (-0,-2) |- +(7.5,-1.5);

\path (G)+(1,-2.8) node[rectangle,draw,minimum width=37mm, minimum height=18mm, dashed](N) {};
\path (G)+(4,-4) node(N1) {$\Psi$};

\end{tikzpicture}
	\end{center}
	\caption{{Illustration of the connection between AE-NLM and IQCs.}}
	\label{Fig:NLMDef_IQC}
\end{figure}
By rearranging the interconnection in Figure~\ref{Fig:NLMDef} to Figure~\ref{Fig:NLMDef_IQC}, it is clear that the calculation of the AE-NLM is equivalent to find a filter
\begin{equation*}
	\Psi:\left\{\begin{aligned} x_\Psi(t+1)&=A_\Psi x_\Psi(t)+B_{\Psi u} u(t),x_\Psi(0)=0\\
	y_\Psi (t)&=C_\Psi x_\Psi(t)+D_{\Psi u}u(t)\\
	r(t)&=\begin{bmatrix}y(t)-y_\Psi (t)\\u(t)
	\end{bmatrix}\end{aligned}\right. 
\end{equation*}
and the minimal $\gamma>0$ that satisfy the time domain IQC \eqref{IQCIneq} with $M=\text{diag}(-\frac{1}{\gamma} I_{n_y},\gamma I_{n_u})$, which corresponds to the supply rate \eqref{SupplyRate2} for $\ell_2$-gains. Thereby, the {AE-NLM} is equal to $\gamma$ and the linear approximation model $G^*$ is the LTI system with system matrices $A_\Psi,B_{\Psi u},C_{\Psi}$, and $D_{\Psi u}$. 

\subsection{Problem formulation}\label{ProblemSetup}

In the previous subsection, we supposed that the nonlinear input-output behaviour $H$ is described by the general nonlinear state-space representation \eqref{NLsys}. However, even the computation of the $\ell_2$-gain of a general nonlinear system is computationally challenging. Therefore, we study throughout the paper the nonlinear discrete-time system \eqref{NLsys} with polynomial dynamics
\begin{equation}\label{NLsysPoly}
	f\in\mathbb{R}[x,u]^{n_x}, h\in\mathbb{R}[x,u]^{n_y} 
\end{equation}
and $f(0,0)=0$ and $h(0,0)=0$, i.e., $x=0$ is a stable equilibrium point. This kind of nonlinear systems is computationally appealing as we can determine system theoretic properties by means of SOS optimization where the square matricial representation \cite{SOSDecomp} of SOS matrices is exploited to conclude on the SOS property via the feasibility of LMIs. We also suppose that the system is operated in the invariant set
\begin{equation}\label{Constraints}
	\begin{aligned}
	\mathbb{P}=\{(x,u)\in\mathbb{R}^{n_x}\times\mathbb{R}^{n_u}: p_i(x,u)\leq 0,\  p_i\in\mathbb{R}[x,u],&\\ i=1,\dots,n_P\}&
	\end{aligned}
\end{equation}
with $(0,0)\in\mathbb{P}$.  \\\indent
The goal of this paper is a framework to calculate an upper bound on the AE-NLM and to determine optimal IQCs for polynomial system~\eqref{NLsys} within \eqref{Constraints} by computationally tractable conditions and without identifying an explicit model but from noisy input-state data. While the verification of dissipativity \eqref{dissipativityInqu} for polynomial systems from data is pursued in \cite{MartinDissi}, the computation of NLMs for polynomial systems has not been analyzed yet, even for known systems.\\\indent 
In order to infer on the polynomial system dynamics \eqref{NLsysPoly} from finitely many input-state samples, we assume to known a vector of distinct monomials $z\in\mathbb{R}[x,u]^{n_z}$ with $z(0,0)=0$ that includes at least all monomials of $f$ and $h$. The knowledge on $z$ requires to some extent insight into the system as exemplary an upper bound on the degree of $f$ and $h$. While the coefficients of $f$ are unidentified, the coefficients of $h$ are supposed to be known which is conceivable due to the access of state measurements. Thus, the output $y$ is defined for the sake of characterization of input-output properties. If only input-output data are available, then the presented framework can be applied for the extended state vector with monomials of inputs and outputs of previous time steps which corresponds to a truncated Volterra series and is analogous to the linear case \cite{AnneDissi}. Summarized, the system dynamics~\eqref{NLsys} with \eqref{NLsysPoly} can be represented by
\begin{align*}
	f(x,u)&=F^*z(x,u)\\
	h(x,u)&=H^*z(x,u),
\end{align*}
where $F^*\in\mathbb{R}^{n_x\times n_z}$ contains the true unidentified coefficients whereas $H^*\in\mathbb{R}^{n_y\times n_z}$ is known. {Since $z$ contains linear independent elements, $F^*$ and $H^*$ are unique.}\\\indent  
To conclude on the unknown matrix $F^*$, we assume the access to input-state data in the presence of noise, i.e.,
\begin{equation}\label{DataSet}
	\{(\tilde{x}_i^+,\tilde{x}_i,\tilde{u}_i)_{i=1,\dots,S}\}
\end{equation}
with $\tilde{x}_i^+=f(\tilde{x}_i,\tilde{u}_i)+\tilde{d}_i$ and {unknown} perturbation $\tilde{d}_i$. Since we examine the NLM of the unperturbed system dynamics and we suppose that the state measurements are affected by additive noise, i.e., $\tilde{x}_i=x_i+d_i$ and $\tilde{x}_i^+=x_i^++d_i^+$ with measurement noise $d_i$ and $d_i^+$, respectively, and the true states $x_i$ and $x_i^+=f(x_i,\tilde{u}_i)$, respectively, {it} holds $\tilde{d}_i=d_i^++f(x_i,\tilde{u}_i)-f(x_i+d_i,\tilde{u}_i)$. Thus, $\tilde{d}_i$ summarizes the additive noise $d_i^+$ and, analogously to \cite{Milanese}, the error when applying the dynamics at the uncertain state $\tilde{x}_i$ instead of the true state $x_i$. Analogously, we can proceed for perturbed inputs $\tilde{u}_i$. Furthermore, if the underlying system~\eqref{NLsys} is influenced by additive process noise then this also has to be considered in the examination of the AE-NLM which would be conceivable as we apply techniques from robust control. However, this will not be within the scope of this paper. \\\indent
In order to conclude on the unidentified parameters $F^*$, we additionally assume that $\tilde{d}_i,i=1,\dots,S,$ are bounded explicitly in each time step as in \cite{MartinDissi}. 
\begin{assu}[Pointwise bounded noise]\label{Noise1}
For the measured data \eqref{DataSet}, suppose for $i=1,\dots,S$ that $\tilde{d}_i\in\mathcal{D}_i$ for compact sets
\begin{equation}\label{SepPrimnoise}
	\mathcal{D}_i=\left\{d\in\mathbb{R}^{n_x}:\begin{bmatrix}1\\d\end{bmatrix}^T\varDelta_{i}\begin{bmatrix}1\\d\end{bmatrix}\leq0 \right\}
\end{equation} 
with invertible matrix $\varDelta_{i}=\begin{bmatrix}\varDelta_{1,i} & \varDelta_{2,i}\\ \varDelta_{2,i}^T & \varDelta_{3,i}\end{bmatrix}$ and $\varDelta_{3,i}\succ0$.
\end{assu}\vspace{0.1cm}

This characterization incorporates disturbances with bounded amplitude $\tilde{d}_i^T\tilde{d}_i-\epsilon^2\leq0$ and {disturbances} that exhibit a fixed signal-to-noise-ratio $\tilde{d}_i^T\tilde{d}_i-\tilde{\epsilon}^2\tilde{x}_i^T\tilde{x}_i\leq0$. Note that deterministic disturbance descriptions are not only frequently supposed in data-driven control \cite{Groningen} and system analysis \cite{AnneDissi} but also in set-membership identification \cite{Milanese}, adaptive control \cite{AdaptiveC}, and robust model predictive control \cite{RMPC}, which are all also successfully applied in practice. If the disturbance is, e.g., Gaussian distributed, then we can still use a bound \eqref{SepPrimnoise} with a certain confidence. 

\section{Data-based set-membership for unidentified coefficient matrices}\label{Supersets}

This section presents a set-membership for $F^*$ by all coefficients matrices that explain the data \eqref{DataSet} for pointwise bounded noise \eqref{SepPrimnoise} which is the basis to determine system properties without identifying an explicit model in Section~\ref{SecGeneralSOS}. A detailed investigation of the accuracy and asymptotic consistency of this set-membership as well as a comparison to the set-membership in \cite{Groningen} are provided in Section~\ref{DataPre2}.\\\indent
At first, we specify analogously to \cite{MartinDissi} the set of all systems
\begin{align}\label{Systemdescribtion}
	x(t+1)=Fz(x(t),u(t))
\end{align}
with coefficients $F\in\mathbb{R}^{n_x\times n_z}$ explaining the data \eqref{DataSet}. 
\begin{defn}[Feasible system set]\label{DefFSS}
The set of all systems \eqref{Systemdescribtion} admissible with the measured data \eqref{DataSet} for pointwise bounded noise \eqref{SepPrimnoise} is given by the feasible system set $\text{FSS}=\{Fz\in\mathbb{R}[x,u]^{n_x}:F\in\Sigma	\}$ with $\Sigma= \{F\in\mathbb{R}^{n_x\times n_z}:\exists \tilde{d}_i\in\mathcal{D}_i\text{\ satisfying\ } \tilde{x}_i^+=Fz(\tilde{x}_i,\tilde{u}_i)+\tilde{d}_i,i=1,\dots,S\}.$
\end{defn}

The feasible system set $\text{FSS}$ is a set-membership representation of the dynamics of the ground-truth system~\eqref{NLsys} with \eqref{NLsysPoly} as $f$ is an element of $\text{FSS}$. Indeed, the samples \eqref{DataSet} suffice $\tilde{x}_i^+=f(\tilde{x}_i,\tilde{u}_i)+\tilde{d}_i$ with $\tilde{d}_i\in\mathcal{D}_i$, and thereby $f\in\text{FSS}$ and $F^*\in\Sigma$. To apply robust control techniques to infer on system properties in the subsequent sections, we require a characterization of the set of admissible coefficients $\Sigma$ of the form  
\begin{align}\label{FSSgen}
	\Sigma_F=\left\{F:\begin{bmatrix}I_{n_z}\\F\end{bmatrix}^T\varDelta_{*i}\begin{bmatrix}I_{n_z}\\F\end{bmatrix}\preceq0,i=1,\dots,n_S\right\},
\end{align}
where the calculation of $\varDelta_{*i}\in\mathbb{R}^{(n_z+n_x)\times (n_z+n_x)}$ from data is shown in the remaining of this section.\\\indent 
We start with an equivalent data-based representation of $\Sigma$ depending on $F^T$.
\begin{lem}[Dual characterization of $\Sigma$]\label{DualSigma}
$\Sigma$ is equivalent to
\begin{align}\label{Sigma_dual}
	\left\{F:\begin{bmatrix}F^T\\I_{n_x}\end{bmatrix}^T\Delta_{i}\begin{bmatrix}F^T\\I_{n_x}\end{bmatrix}\preceq0,i=1,\dots,S\right\}
\end{align}
with the data-depende{nt} matrices
\begin{align*}
	&\Delta_{i}=\\
	&\begin{bmatrix}-\tilde{z}_i\Delta_{1,i}\tilde{z}_i^T & \tilde{z}_i(\Delta_{1,i}\tilde{x}_i^{+^T}-\Delta_{2,i})\\(\tilde{x}_i^+\Delta_{1,i}-\Delta_{2,i}^T)\tilde{z}_i^T & \begin{bmatrix}\tilde{x}_i^{+^T}\\I_{n_x}\end{bmatrix}^T\begin{bmatrix}-\Delta_{1,i} & \Delta_{2,i}\\ \Delta_{2,i}^T & -\Delta_{3,i}\end{bmatrix}\begin{bmatrix}\tilde{x}_i^{+^T}\\I_{n_x}\end{bmatrix} \end{bmatrix},
\end{align*}
$\tilde{z}_i=z(\tilde{x}_i,\tilde{u}_i)$, and $\begin{bmatrix}\Delta_{1,i} & \Delta_{2,i}\\ \Delta_{2,i}^T & \Delta_{3,i}\end{bmatrix}=\varDelta_{i}^{-1}$.
\end{lem}\vspace{0.1cm}
\begin{proof}
By the dualization lemma~\cite{SchererLMI}, the noise bounds from \eqref{SepPrimnoise} are equivalent to the dual form 	
\begin{align}\label{DualNoise}
	\begin{bmatrix}d^T\\I_{n_x}\end{bmatrix}^T\begin{bmatrix}-\Delta_{1,i} & \Delta_{2,i}\\ \Delta_{2,i}^T & -\Delta_{3,i}\end{bmatrix}\begin{bmatrix}d^T\\I_{n_x}\end{bmatrix}\preceq0, \Delta_{1,i}<0
\end{align}
where $\varDelta_{i}^{-1}$ exists by Assumption~\ref{Noise1}. Combining the dual version \eqref{DualNoise} of the noise bound, data samples~\eqref{DataSet}, and the system dynamics \eqref{Systemdescribtion} yields the dual representation \eqref{Sigma_dual}.
\end{proof}

To derive \eqref{FSSgen} from \eqref{Sigma_dual}, the dualization lemma can not be employed on the dual representation \eqref{Sigma_dual} as the invertibility of  $\Delta_i\in\mathbb{R}^{(x_x+n_z)\times(x_x+n_z)}$ is violated because its left upper block is rank one, and hence it is not full column rank for $n_z\geq n_x+1$. To attain nevertheless a form as \eqref{FSSgen}, we suggest to first calculate an ellipsoidal outer approximation of \eqref{Sigma_dual} as in \cite{BoydLMI} and then to dualize.

\begin{lrop}[Pointwise superset of $\Sigma$]\label{Primal_elli}
Let $\tilde{Z}=\begin{bmatrix}z(\tilde{x}_1,\tilde{u}_1) & \cdots & z(\tilde{x}_S,\tilde{u}_S)\end{bmatrix}$ be full row rank. Then there exist a positive definite matrix $\Delta_{1\text{p}}\in\mathbb{R}^{n_z\times n_z}$, matrix $\Delta_{2\text{p}}\in\mathbb{R}^{n_z\times n_x}$, and scalars $\alpha_1,\dots,\alpha_S\geq0$ solving\vspace{-0.1cm}	
\begin{align}\label{Elli_LMI}
	\begin{bmatrix}\Delta_{1\text{p}} & \Delta_{2\text{p}} & 0\\ \Delta_{2\text{p}}^T & -I_{n_x} & \Delta_{2\text{p}}^T\\ 0 & \Delta_{2\text{p}} & -\Delta_{1\text{p}} \end{bmatrix}-\sum_{i=1}^{S}\alpha_i\begin{bmatrix}\Delta_i &  0\\  0 & 0\end{bmatrix}\preceq0.
\end{align}	
Moreover, then the set of feasible coefficients $\Sigma$ is a subset of 
\begin{align}\label{Sigma_p}
	\Sigma_{\text{p}}=\left\{F\in\mathbb{R}^{n_x\times n_z}:\begin{bmatrix}I_{n_z}\\F\end{bmatrix}^T\varDelta_{\text{p}}\begin{bmatrix}I_{n_z}\\F\end{bmatrix}\preceq0\right\}
\end{align}
with $\varDelta_{\text{p}}=\begin{bmatrix}-\varDelta_{1\text{p}} & \varDelta_{2\text{p}}\\ \varDelta_{2\text{p}}^T & -\varDelta_{3\text{p}}\end{bmatrix}$, $\begin{bmatrix}\varDelta_{1\text{p}} & \varDelta_{2\text{p}}\\ \varDelta_{2\text{p}}^T & \varDelta_{3\text{p}}\end{bmatrix}=\Delta_\text{p}^{-1}$, and $\Delta_\text{p}=\begin{bmatrix}\Delta_{1\text{p}} & \Delta_{2\text{p}}\\ \Delta_{2\text{p}}^T & \Delta_{2\text{p}}^T\Delta_{1\text{p}}^{-1}\Delta_{2\text{p}}-I_{n_x}\end{bmatrix}$.
\end{lrop}\vspace{0.1cm}
\begin{proof}
First, we show that LMI $\eqref{Elli_LMI}$ has a solution if $\tilde{Z}$ is full row rank by extending Lemma 2 of \cite{Ellipsoid_DePersis2} to general quadratic noise characterizations. To this end, we introduce the abbreviation for the block matrices of $\Delta_i=\begin{bmatrix}\Gamma_{1,i} & \Gamma_{2,i}\\ \Gamma_{2,i}^T & \Gamma_{3,i}\end{bmatrix}$. Moreover, let $\alpha\geq0$ be a to-be-optimized scalar and set $\alpha_i=-\frac{\alpha}{\Delta_{1,i}}$, $\Delta_{1\text{p}}=\alpha \tilde{Z}\tilde{Z}^T$, and $\Delta_{2\text{p}}=- \sum_{i=1}^{S}\frac{\alpha}{\Delta_{1,i}}\Gamma_{2,i}$. Note that this choice is valid as $\Delta_{1,i}<0$ by \eqref{DualNoise} and $\Delta_{1\text{p}}\succ0$ by the full row rank of $\tilde{Z}$. By $\tilde{Z}\tilde{Z}^T=\sum_{i=1}^{S} z(\tilde{x}_i)z(\tilde{x}_i)^T$ together with the choice of $\Delta_{1\text{p}}$ and $\Delta_{2\text{p}}$, the first block row and first block column of \eqref{Elli_LMI} are zero, and thus \eqref{Elli_LMI} is satisfied if  
\begin{align}\label{LMI2}
	\begin{bmatrix}-I_{n_x}+\sum_{i=1}^{S}\frac{\alpha}{\Delta_{1,i}}\Gamma_{3,i} & -\sum_{i=1}^{S}\frac{\alpha}{\Delta_{1,i}}\Gamma_{2,i}^T\\ -\sum_{i=1}^{S}\frac{\alpha}{\Delta_{1,i}}\Gamma_{2,i} & -\alpha \tilde{Z}\tilde{Z}^T \end{bmatrix}\preceq0.
\end{align}
The full row rank of $\tilde{Z}$ implies $\tilde{Z}\tilde{Z}^T\succ0$, and hence \eqref{LMI2} is satisfied if $I_{n_x}-\alpha\sum_{i=1}^{S}\frac{1}{\Delta_{1,i}}\Gamma_{3,i}-\alpha(\sum_{i=1}^{S}\frac{1}{\Delta_{1,i}}\Gamma_{2,i}^T)(\tilde{Z}\tilde{Z}^T)^{-1}(\sum_{i=1}^{S}\frac{1}{\Delta_{1,i}}\Gamma_{2,i})\succeq0$ by the Schur complement. Finally, this holds if $\alpha>0$ is chosen small enough.\\
Next, we show that $\Sigma\subseteq{\Sigma}_\text{p}$ by adapting \cite{BoydLMI} (Chapter 3.7.2) for matrix ellipsoidal outer approximation. Since we can find a $\Delta_{1\text{p}}\succ0$ solving \eqref{Elli_LMI}, the Schur complement yields for \eqref{Elli_LMI} the equivalent condition 
\begin{align*}
	\begin{bmatrix}\Delta_{1\text{p}} & \Delta_{2\text{p}}\\ \Delta_{2\text{p}}^T & \Delta_{2\text{p}}^T\Delta_{1\text{p}}^{-1}\Delta_{2\text{p}}-I_{n_x}\end{bmatrix}-\sum_{i=1}^{S}\alpha_i\Delta_{i}\preceq0.
\end{align*}
Multiplying this inequality by the matrix $\begin{bmatrix}F^T\\I_{n_x}\end{bmatrix}$ from the right-hand side and its transpose from the left-hand side and applying the S-procedure yield that
\begin{align}\label{DualElli}
	\begin{bmatrix}F^T\\I_{n_x}\end{bmatrix}^T\Delta_\text{p}\begin{bmatrix}F^T\\I_{n_x}\end{bmatrix}\preceq0	
\end{align}
holds for all $F\in\mathbb{R}^{n_x\times n_z}$ satisfying $\begin{bmatrix}F^T\\I_{n_x}\end{bmatrix}^T\Delta_{i}\begin{bmatrix}F^T\\I_{n_x}\end{bmatrix}\preceq0,i=1,\dots,S,$ where latter is equivalent to $\Sigma$ by Lemma~\ref{DualSigma}. Applying the dualization lemma on \eqref{DualElli} yields \eqref{Sigma_p}, and thus $\Sigma\subseteq{\Sigma}_\text{p}$.\\
It remains to show the invertibility of $\Delta_\text{p}$. For that purpose, suppose $\Delta_\text{p}$ is not full rank, then there exists a vector $r=\begin{bmatrix}r_1^T& r_2^T\end{bmatrix}^T\neq0$ such that
\begin{align*}
	\Delta_p\begin{bmatrix}r_1\\ r_2\end{bmatrix}=\begin{bmatrix}\Delta_{1\text{p}}r_1+\Delta_{2\text{p}}r_2\\ \Delta_{2\text{p}}^T(r_1+\Delta_{1\text{p}}^{-1}\Delta_{2\text{p}}r_2)-r_2\end{bmatrix}=0.
\end{align*}
Thus, $\Delta_{2\text{p}}^T(r_1-\Delta_{1\text{p}}^{-1}\Delta_{1\text{p}}r_1)-r_2=-r_2=0$. Together with $\Delta_{1\text{p}}\succ0$, the first equation implies $r_1=0$, and hence $r=0$. Due to the contradiction with $r\neq0$, $\Delta_\text{p}$ is full rank.	
\end{proof}

Using the S-procedure in the proof yields a sufficient but not necessary condition. Hence, \eqref{DualElli} is a superset of \eqref{Sigma_dual} and $\Sigma_\text{p}$ is not a tight characterization of $\Sigma$. Moreover, we assess full row rank of $\tilde{Z}$ to be not restrictive as it can be achieved by increasing the number of columns by means of additional samples and the rank condition can easily be checked from data. This rank condition is also not surprising as it corresponds to a persistence of excitation condition \cite{Ellipsoid_DePersis2}.\\\indent
Geometrically, we compute in Proposition~\ref{Primal_elli} an ellipsoidal outer approximation \eqref{DualElli} of the intersection of quadratic matrix inequalities \eqref{Sigma_dual}, where each describes the unbounded space between two parallel hyperplanes. Under the full row rank of $\tilde{Z}$, this intersection, i.e., $\Sigma$, is bounded. According to \cite{BoydLMI}, we can derive the outer approximating ellipsoid with minimal volume by minimizing over the convex function $\log(\det(\Delta_{1\text{p}}^{-1}))$ or with minimal diameter by maximizing over $\kappa>0$ with $\Delta_{1\text{p}}\succeq\kappa I_{n_z}$ using {SDP}. By this procedure, the volume or diameter, respectively, of $\Sigma_\text{p}$ is monotonically decreasing if the data set \eqref{DataSet} is extended by additional samples. \\\indent
In Section~\ref{DataPre2}, we provide two additional supersets of $\Sigma$, which are more conservative than $\Sigma_\text{p}$ but do not call for solving an LMI, and thus are interesting if a large amount of samples are available. Moreover, Section~\ref{DataPre2} shows the asymptotic consistency of all three supersets. Note that the results from Section~\ref{DataPre2} are not required for the data-driven determination of system properties in Section~\ref{SecGeneralSOS} and \ref{Ext}.

\section{Data-driven inference on AE-NLM}\label{SecGeneralSOS}

In this section, we treat the derivation of an SDP to calculate from the data-based superset $\Sigma_\text{p}$ a guaranteed upper bound on the AE-NLM and the \textquoteleft optimal\textquoteright\ linear approximation of the unidentified system \eqref{NLsys} with \eqref{NLsysPoly}. This framework will then be extended in Section~\ref{Ext} to deduce analogously SDPs to determine optimal IQCs.\\\indent 
Consider the problem setup in Section~\ref{ProblemSetup} and a data-driven inference on the unidentified coefficients $F^*$ of the form \eqref{FSSgen} where $\Sigma_F$ is interchangeable by the superset $\Sigma_\text{p}$ or the supersets $\Sigma_\text{w}$ and $\Sigma_\text{c}$ defined in Section~\ref{DataPre2}. For the computation of the AE-NLM, let the set of stable linear systems $\mathcal{G}$ be described by LTI systems
\begin{equation}\label{LinearSystem}
	G:\left\{\begin{aligned} x_\Psi(t+1)&=A_\Psi x_\Psi(t)+B_{\Psi} u(t),x_\Psi(0)=0\\
	r(t)&=C_\Psi x_\Psi(t)+D_{\Psi}u(t)\end{aligned}\right. 
\end{equation}
with $A_\Psi\in\mathbb{R}^{n_x\times n_x}$, $B_{\Psi}\in\mathbb{R}^{n_x\times n_u}$, $C_\Psi\in\mathbb{R}^{n_y \times n_x}$, and $D_{\Psi}\in\mathbb{R}^{n_y \times n_u}$. Since $G$ will be designed such that the interconnection in Figure~\ref{Fig:NLMDef} is $\ell_2$-gain stable with stable nonlinear system $H$, $A_\Psi$ will implicitly be Schur. \\\indent
The key idea to determine input-output properties from the set-membership representation $\Sigma_F$ of the true unidentified coefficients $F^*$, i.e. $F^*\in\Sigma_F$, relies on the fact that the ground-truth system \eqref{NLsys} with \eqref{NLsysPoly} exhibits a certain input-output property if all systems of the feasible system set $\text{FSS}=\{Fz\in\mathbb{R}[x,u]^{n_x}:F\in\Sigma_F\}$ exhibit this input-output property. Therefore, we can provide a data-based criterion to verify the AE-NLM with a given linear surrogate model for the polynomial system.
\begin{lem}[Data-driven verification of AE-NLM]\label{ThmAnaNLM}
Let the data samples \eqref{DataSet} satisfy Assumption~\ref{Noise1} and let a scalar $\Phi>0$ and a stable LTI system \eqref{LinearSystem} be given. Then the AE-NLM of the polynomial system \eqref{NLsys} with \eqref{NLsysPoly} within the operation set \eqref{Constraints} is upper bounded by $\Phi$ if there exist a matrix $\mathcal{X}\succ0$, non-negative scalars $\tau_{\Sigma1},\dots,\tau_{\Sigma n_S}$, and polynomials $t_i\in\text{SOS}[x,u],i=1,\dots,n_P$ such that $\psi\in\text{SOS}[x,x_\Psi,u,\text{vec}(F)]$ with 
\begin{equation}\label{DissiAENLM}
	\begin{aligned}
	\psi=&\begin{bmatrix}x\\x_\Psi\end{bmatrix}^T\mathcal{X}\begin{bmatrix}x\\x_\Psi\end{bmatrix}-\begin{bmatrix}Fz\\A_\Psi x_\Psi+B_\Psi u\end{bmatrix}^T\mathcal{X}\begin{bmatrix}Fz\\A_\Psi x_\Psi+B_\Psi u\end{bmatrix}\\
	&{+}\Phi u^Tu-\frac{1}{\Phi}e^Te+\sum_{i=1}^{n_S}\tau_{\Sigma i}\begin{bmatrix}z\\Fz\end{bmatrix}^T\varDelta_{*i}\begin{bmatrix}z\\Fz\end{bmatrix}+\sum_{i=1}^{n_P}p_it_i
	\end{aligned}
\end{equation}	
and $e(x,x_\Psi,u)=H^*z(x,u)-C_\Psi x_\Psi-D_\Psi u$.
\end{lem}
\begin{proof}
Consider the interconnection of error system $\Delta=H-G$ in Figure~\ref{Fig:NLMDef} with state-space representation
\begin{equation}\label{ErrorSS}
\begin{aligned}
	\begin{bmatrix}x(t+1)\\x_\Psi(t+1)\end{bmatrix}&=\begin{bmatrix}F^*z(x(t),u(t))\\A_\Psi x_\Psi(t)+B_\Psi u(t)\end{bmatrix},\begin{bmatrix}x(0)\\x_{\Psi}(0)\end{bmatrix}=0\\
	e(t)&=H^*z(x(t),u(t))-C_\Psi x_\Psi(t)-D_\Psi u(t)
\end{aligned}.
\end{equation}
Since the AE-NLM of $H$ is equal to the $\ell_2$-gain of $\Delta:u\mapsto e$, $\Phi$ is an upper bound of the AE-NLM by Proposition~\ref{GainSys} if \eqref{ErrorSS} is dissipative on $(x,u,x_\Psi)\in\mathbb{P}\times\mathbb{R}^{n_x}$ with respect to the supply rate $s(e,u)=\Phi ||u||_2^2-\frac{1}{\Phi}||e||_2^2$. By Definition~\ref{DissiDef}, this holds true if there exists a storage function $\lambda(x,x_\psi)=\begin{bmatrix}x^T & x_\Psi^T\end{bmatrix}\mathcal{X}\begin{bmatrix}x^T&x_\Psi^T\end{bmatrix}^T$ such that for all $(x,u,x_\Psi)\in\mathbb{P}\times\mathbb{R}^{n_x}$
\begin{equation}\label{Inequ}
	0\leq s(e,u)+\begin{bmatrix}x\\x_\Psi\end{bmatrix}^T\mathcal{X}\begin{bmatrix}x\\x_\Psi\end{bmatrix}-\star^T\mathcal{X}\begin{bmatrix}F^*z\\A_\Psi x_\Psi+B_\Psi u\end{bmatrix},
\end{equation}
where $\star$ is a placeholder for the matrix on the right. Since the true coefficient matrix $F^*$ is unknown but $F^*\in\Sigma_F$, we require that \eqref{Inequ} holds for all $F\in\Sigma_F$. Therefore, we require the generalized S-procedure for polynomials which follows from the Positivstellensatz \cite{ProofProp} (Lemma 2.1): a polynomial $q\in\mathbb{R}[v]$ is non-negative on $\{v\in\mathbb{R}^{n_v}:c_1(v)\leq0,\dots,c_k(v)\leq0\}$ if there exist polynomials $q_i\in\text{SOS}[v],i=1,\dots,k,$ such that $q(v)+\sum_{i=1}^{k}q_i(v)c_i(v)\geq0,\forall v\in\mathbb{R}^{n_v}$. Together with $F\in\Sigma_F$ implying that for all $(x,u)\in\mathbb{R}^{n_x}\times\mathbb{R}^{n_u}$ and $i=1,\dots,n_S$
\begin{equation*}
	z(x,u)^T\begin{bmatrix}I_{n_z}\\F\end{bmatrix}^T\varDelta_{*i}\begin{bmatrix}I_{n_z}\\F\end{bmatrix}z(x,u)\leq0,
\end{equation*}
we conclude that the dissipativity criterion \eqref{Inequ} holds for all $F\in\Sigma_F$ if there exist a $\mathcal{X}\succ0$, scalars $\tau_{\Sigma1},\dots,\tau_{\Sigma n_S}\geq0$, and polynomials $t_i\in\text{SOS}[x,u],i=1,\dots,n_P,$ such that $\psi(x,x_\Psi,u,\text{vec}(F))\geq0$ for all $(x,x_\Psi,u,\text{vec}(F))\in\mathbb{R}^{n_x}\times\mathbb{R}^{n_x}\times\mathbb{R}^{n_u}\times\mathbb{R}^{n_xn_z}$ which is implied by 
$\psi\in\text{SOS}[x,x_\Psi,u,\text{vec}(F)]$ due to the relaxation that any SOS polynomial is non-negative. 
\end{proof}

Lemma~\ref{ThmAnaNLM} constitutes a computationally tractable SOS condition to verify an upper bound of the AE-NLM from noisy input-state data if a linear approximation model is given. In fact, $\psi$ is a polynomial in $\mathbb{R}[x,x_\Psi,u,\text{vec}(F)]$ and linear in the optimization variables $\mathcal{X}, \tau_{\Sigma i},$ and $\tau_i$, and hence we can check by standard SOS solvers \cite{YALMIP} whether $\psi$ is an SOS polynomial. Furthermore, Lemma~\ref{ThmAnaNLM} is a special case of Theorem 1 in \cite{MartinDissi} where the data-based dissipativity verification for polynomial systems regarding polynomial supply rates is investigated using polynomial storage functions and noise specifications. Here, we require quadratic storage functions and non-positive scalars $\tau_{\Sigma i}$ instead of SOS polynomials in order to attain LMIs in the following.\\\indent
Since the linear approximation model, defined by $A_\Psi, B_\Psi, C_\Psi,$ and $D_\Psi$ in Lemma~\ref{ThmAnaNLM}, is usually not available and appears non-convex in \eqref{DissiAENLM}, we deduce an equivalent condition to \eqref{DissiAENLM} which is linear in the to-be-optimized variables. 
\begin{thm}[Data-driven inference on AE-NLM]\label{SynAENLM}
Suppose the data samples \eqref{DataSet} suffice Assumption~\ref{Noise1} and the vector $z$ contains $x$ and $u$, i.e., there exist matrices $T_x\in\mathbb{R}^{n_x\times n_z}$ and $T_u\in\mathbb{R}^{n_u\times n_z}$ with $x=T_x z$ and $u=T_u z$, respectively. If there exist matrices $X,Y^{-1}\succ0$, non-negative scalars $\tau_{\Sigma 1},\dots, \tau_{\Sigma n_S}$, and $\tau_x$, $\Phi>0$, matrices $\tilde{K}\in\mathbb{R}^{n_x\times n_x},L\in\mathbb{R}^{n_x\times n_u},\tilde{M}\in\mathbb{R}^{n_y\times n_x},N\in\mathbb{R}^{n_y\times n_u}$, and polynomials $z_i\tau_i\in\text{SOS}[x,u],i=1,\dots,n_P$, with a vector of monomials $z_i\in\mathbb{R}[x,u]^{1\times\beta}$, to-be-optimized coefficients $\tau_i\in\mathbb{R}^\beta$, and a linear mapping $P_i:\mathbb{R}^\beta\rightarrow\mathbb{R}^{n_z \times n_z}$ with
\begin{align}\label{QuadraticDecomp}
	z_i\tau_ip_i = z^TP_i(\tau_i)z,
\end{align}	
satisfying 
\begin{equation}\label{TransCond1}
	\mathcal{Y}^T\mathcal{X}\mathcal{Y}:=\begin{bmatrix}Y^{-1}&Y^{-1}\\Y^{-1}&X\end{bmatrix}\succ0
\end{equation}
and \eqref{LMICond} with 
\begin{equation*}
	\Omega = \begin{bmatrix}\begin{array}{c|c}\begin{matrix}0 & 0\\ \tilde{K} & 0\end{matrix} & \begin{matrix}0 & Y^{-1}\\ LT_u & X \end{matrix}\\\hline \begin{matrix}-\tilde{M}\phantom{\Big|}& 0\end{matrix} & \begin{matrix} H^*-NT_u & 0	\end{matrix}\end{array}
	\end{bmatrix},
\end{equation*}
then the AE-NLM of the ground-truth {polynomial} system \eqref{NLsys} {with \eqref{NLsysPoly}} is upper bounded by $\Phi$ for the linear approximation model \eqref{LinearSystem} with $A_\Psi$, $B_\Psi$, $C_\Psi$, and $D_\Psi$ from 
\begin{equation}\label{InvTrafo1}
	\begin{bmatrix}K & L\\ M & N\end{bmatrix}=\begin{bmatrix}U & 0\\ 0& I\end{bmatrix}\begin{bmatrix}A_\Psi & B_\Psi\\ C_\Psi & D_\Psi\end{bmatrix}\begin{bmatrix}V^T & 0\\ 0 & I\end{bmatrix}
\end{equation}
with $K=\tilde{K}Y$, $M=\tilde{M}Y$, and $I_{n_x}-XY=UV^T$.	
\end{thm}
\begin{figure*}
	\begin{align}\label{LMICond}
	&0\preceq\begin{bmatrix}
	\star^T\begin{bmatrix}\begin{array}{c|c|c|c|c}
	\mathcal{Y}^T\mathcal{X}\mathcal{Y} & 0 &0 & 0 & 0\\\hline 0 & \Phi I_{n_u}  & 0 &0&0 \\\hline	
	0 & 0 & \sum_{i=1}^{n_S}\tau_{\Sigma i}\varDelta_{*i}\phantom{\Big|} & 0 & 0\\\hline					
	0 & 0 & 0 & \sum_{i=1}^{n_P}P_i(\tau_i)\phantom{\Big|} & 0\\\hline 0& 0 & 0 & 0 & \tau_x I_{n_x} 
	\end{array}	\end{bmatrix}
	\begin{bmatrix}\hspace{0.2cm}\begin{matrix}	I_{2n_x} & 0 & 0 \\	\hline 
	0 & T_u & 0\\\hline 0 & I_{n_z} & 0 \\ 0 & 0 & I_{n_x}\\\hline 
	0 & I_{n_z} & 0\\\hline \begin{bmatrix}I_{n_x} & I_{n_x}\end{bmatrix}\phantom{\Big|} & -T_x & 0\end{matrix}\hspace{0.2cm}\end{bmatrix}
	& \Omega^T\\
	\Omega & \begin{bmatrix}\mathcal{Y}^T\mathcal{X}\mathcal{Y}  & 0 \\ 0&  \Phi I_{n_y} \end{bmatrix}
	\end{bmatrix}
	\end{align}
\end{figure*}
\begin{proof}
Retain from Lemma~\ref{ThmAnaNLM} the condition that $\psi$ has to be non-negative for all $(x,x_\Psi,u,\text{vec}(F))\in\mathbb{R}^{n_x}\times\mathbb{R}^{n_x}\times\mathbb{R}^{n_u}\times\mathbb{R}^{n_xn_z}$. Instead of applying an SOS relaxation as in Lemma~\ref{ThmAnaNLM}, we require that $\psi$ is non-negative for all $x,x_\Psi,u$, and $F$ with $\begin{bmatrix}\begin{bmatrix}I_{n_x} & 0  \end{bmatrix}& -T_x & 0\end{bmatrix}\phi=0$ and $\phi=\begin{bmatrix}\begin{bmatrix}x^T & x_\Psi^T  \end{bmatrix}  & z^T & z^TF^T\end{bmatrix}^T$. Since $\psi$ is a homogeneous quadratic polynomial in $\phi$, Finsler's lemma yields the equivalent condition \eqref{LMIAna} 
\begin{figure*}
	\vspace{-0.8cm}
	\begin{align}\label{LMIAna}
	&0\preceq E_1^T
	\text{diag}\left(\mathcal{X},-\mathcal{X}\Bigg|\Phi I_{n_u},-\frac{1}{\Phi}I_{n_y}\Bigg|\sum_{i=1}^{n_S}\tau_{\Sigma i}\varDelta_{*i}\Bigg|\sum_{i=1}^{n_P}P_i(\tau_i)\Bigg|\tau_x I_{n_x}\right)E_1
	\end{align}	
\end{figure*}
with $\tau_x\geq0$,
\begin{align*}
	E_1&=\begin{bmatrix}\hspace{0.2cm}\begin{matrix}	I_{2n_x} & 0 & 0 \\ \mathcal{A} & \mathcal{B}_z & \mathcal{B}_{Fz}\\	\hline 
	0 & T_u & 0\\  \mathcal{C} & \mathcal{D}_z & \mathcal{D}_{Fz}\\\hline 0 & I_{n_z} & 0 \\ 0 & 0 & I_{n_x}\\\hline 
	0 & I_{n_z} & 0\\\hline \begin{bmatrix}I_{n_x} & 0\end{bmatrix}\phantom{\Big|} & -T_x & 0\end{matrix}\hspace{0.2cm}\end{bmatrix},
\end{align*}
\begin{align*}
	\begin{bmatrix}\hspace{0.15cm}\begin{matrix}
	x(t+1)\\x_\Psi(t+1)\\\hline e(t)\end{matrix}\hspace{0.15cm}\end{bmatrix}&=\begin{bmatrix}\begin{array}{c|c}\begin{matrix}0 & 0\\ 0 & A_\Psi\end{matrix}  & \begin{matrix} 0 & I\\ B_\Psi T_u & 0\end{matrix}\\\hline \begin{matrix}0 & -C_\Psi\end{matrix} & \begin{matrix} H^*-D_\Psi T_u & 0	\end{matrix}\end{array}
	\end{bmatrix}\phi(t)\\
	&=:\begin{bmatrix}\begin{array}{c|c}\mathcal{A} & \begin{matrix}	\mathcal{B}_z & \mathcal{B}_{Fz}\end{matrix}\\\hline\mathcal{C} & \begin{matrix}\mathcal{D}_z & \mathcal{D}_{Fz}\end{matrix}	\end{array}
	\end{bmatrix}\phi(t).
\end{align*}
Since \eqref{LMIAna} is independent of $x$, we can apply techniques from the linear robust control literature to linearize \eqref{LMIAna} regarding the optimization variables.\\
To this end, define the partition from \cite{SchererLMI} 
\begin{equation*}
	\mathcal{X} = \begin{bmatrix}X & U\\ U^T & *\end{bmatrix}\ \text{and}\ \mathcal{X}^{-1} = \begin{bmatrix}Y & V\\ V^T & *\end{bmatrix}
\end{equation*}
with $XY+UV^T=I_{n_x}$ and the congruence transformation of condition $\mathcal{X}\succ0$ with
\begin{equation*}
	\mathcal{Y}_1=\begin{bmatrix}Y & I_{n_x}\\ V^T & 0\end{bmatrix}
\end{equation*}
which yields 
\begin{equation*}
	\mathcal{Y}_1^T\mathcal{X}\mathcal{Y}_1=\begin{bmatrix}Y & I_{n_x}\\ I_{n_x} & X\end{bmatrix}\succ0.
\end{equation*}
Hence, $I_{n_x}-XY$ is non-singular such that we can factorize $I_{n_x}-XY=UV^T$ with square and non-singular matrices $U,V\in\mathbb{R}^{n_x\times n_x}$. Contrary to \cite{SchererLMI}, we require an additional congruence transformation with 
\begin{equation*}
	\mathcal{Y}_2=\begin{bmatrix}Y^{-1} & 0\\ 0 & I_{n_x}\end{bmatrix}.
\end{equation*}
To apply both congruence transformation in the sequel, we calculate for $\mathcal{Y}=\mathcal{Y}_1\mathcal{Y}_2$
\begin{align*}
	&\begin{bmatrix}\mathcal{Y} & 0\\ 0 & I_{n_y}\end{bmatrix}^T
	\begin{bmatrix}\begin{array}{c|c}\mathcal{X}\mathcal{A} & \begin{matrix}\mathcal{X}\mathcal{B}_z & 	\mathcal{X}\mathcal{B}_{Fz}\end{matrix}\\\hline\mathcal{C} & \begin{matrix}\mathcal{D}_z & \mathcal{D}_{Fz}\end{matrix}	\end{array}
	\end{bmatrix}\begin{bmatrix}\mathcal{Y} & 0 \\ 0 & I_{n_z+n_x} \end{bmatrix}\\
	=&\begin{bmatrix}\mathcal{Y}_2 & 0\\ 0 & I_{n_y}\end{bmatrix}\begin{bmatrix}\begin{array}{c|c}\begin{matrix}0 & 0\\ K & 0\end{matrix} & \begin{matrix}0 & I_{n_x}\\ LT_u & X \end{matrix}\\\hline \begin{matrix}-M & 0\end{matrix} & \begin{matrix} H^*-NT_u & 0	\end{matrix}\end{array}
	\end{bmatrix}\begin{bmatrix}\mathcal{Y}_2 & 0 \\ 0 & I_{n_z+n_x}\end{bmatrix}\\
	=&\ \Omega
\end{align*}
with $\tilde{K}=KY^{-1}$, $\tilde{M}=MY^{-1}$, and
\begin{equation*}
	\begin{bmatrix}K & L\\ M & N\end{bmatrix}=\begin{bmatrix}U & 0\\ 0& I_{n_y}\end{bmatrix}\begin{bmatrix}A_\Psi & B_\Psi\\ C_\Psi & D_\Psi\end{bmatrix}\begin{bmatrix}V^T & 0\\ 0 & I_{n_u}\end{bmatrix}
\end{equation*}
from \cite{SchererLMI} (Section 4.2). Applying the congruence transformation with $\mathcal{Y}$ to $\mathcal{X}\succ0$ yields \eqref{TransCond1} and applying the congruence transformation with $\text{diag}(\mathcal{Y}, I_{n_z}, I_{n_x})$ to \eqref{LMIAna} yields \eqref{LMICon} 
\begin{figure*}
\vspace{-0.6cm}
\begin{align}\label{LMICon}
	&0\preceq E_2^T
	\text{diag}\left(\mathcal{Y}^T\mathcal{X}\mathcal{Y},-\mathcal{Y}^{-1}\mathcal{X}^{-1}\mathcal{Y}^{T^{-1}}\Bigg|\Phi I_{n_u},-\frac{1}{\Phi}I_{n_y}\Bigg|\sum_{i=1}^{n_S}\tau_{\Sigma i}\varDelta_{*i} \Bigg|\sum_{i=1}^{n_P}P_i(\tau_i)\Bigg|\tau_x I_{n_x}	\right)
	E_2
\end{align}
\begin{tikzpicture}
	\draw[-,line width=0.7pt] (0,0) -- (18,0);
\end{tikzpicture}
\vspace{-0.3cm} 
\end{figure*}
with 
\begin{equation*}
	E_2= \begin{bmatrix}\hspace{0.2cm}\begin{matrix}	I_{2n_x} & 0 & 0 \\ \mathcal{Y}^T\mathcal{X}\mathcal{A}\mathcal{Y} & \mathcal{Y}^T\mathcal{X}\mathcal{B}_z & \mathcal{Y}^T\mathcal{X}\mathcal{B}_{Fz}\\	\hline 
	0 & T_u & 0\\  \mathcal{C}\mathcal{Y} & \mathcal{D}_z & \mathcal{D}_{Fz}\\\hline 0 & I_{n_z} & 0 \\ 0 & 0 & I_{n_x}\\\hline 
	0 & I_{n_z} & 0\\\hline \begin{bmatrix}I_{n_x} & I_{n_x}\end{bmatrix} & -T_x & 0\end{matrix}\hspace{0.2cm}\end{bmatrix},
\end{equation*}	
where $\mathcal{Y}$ is invertible as $V$ is invertible. Finally, \eqref{LMICon} is equivalent to \eqref{LMICond} by the Schur complement.
\end{proof}

Before Theorem~\ref{SynAENLM} is employed in a numerical example, some comments are appropriate. First, Lemma~\ref{ThmAnaNLM} is equivalent to Theorem~\ref{SynAENLM} while the matrix inequalities \eqref{TransCond1} and \eqref{LMICond} are linear in the optimization variables $X,Y^{-1},\tau_{\Sigma 1},\dots, \tau_{\Sigma n_S},\tau_x,\Phi,\tilde{K},L,\tilde{M},N,$ and $\tau_1,\dots,\tau_{n_P}$. Thus, the smallest guaranteed upper bound on the AE-NLM can be computed by minimizing over $\Phi$ subject to the LMI conditions \eqref{TransCond1} and \eqref{LMICond} and the SOS conditions on the polynomials $z_i\tau_i,i=1,\dots,n_P,$ which boil down to an LMI condition by the square matricial representation \cite{SOSDecomp}. Secondly, since $Y^{-1}$ is nonsingular, we can compute square and nonsingular matrices $U$ and $V$ by a matrix factorization to perform the inverse transformation from $\tilde{K},L,\tilde{M},$ and $N$ to $A_\Psi, B_\Psi, C_\Psi,$ and $D_\Psi$ which constitute the \textquoteleft optimal\textquoteright\ linear approximation model of $H$. We summarize the calculation of AE-NLM and the \textquoteleft optimal\textquoteright\ linear approximation from data in the following algorithm.

\begin{algo}[Data-driven inference on AE-NLM from $\mathit\Sigma_\text{p}$]\label{Algorithm1}\indent
	\vspace{-0.4cm}\begin{itemize}	
		\item[$0.)$] Given the vector $z$ and data \eqref{DataSet} that satisfy Assumption~\ref{Noise1}. 
		\item[$1.)$] Compute $\Delta_i$ from Lemma~\ref{DualSigma}, solve LMI~\eqref{Elli_LMI}, and compute $\Sigma_\text{p}$ from \eqref{Sigma_p}.
		\item[$2.)$] Solve the SDP in Theorem~\ref{SynAENLM}, i.e., minimize $\Phi>0$ subject to \eqref{TransCond1} and \eqref{LMICond}. AE-NLM is upper bounded by $\Phi$.		
		\item[$3.)$] Calculate $U,V$ from, e.g., a singular value decomposition of $I_{n_x}-XY$. Derive $A_\Psi, B_\Psi, C_\Psi,$ and $D_\Psi$ from \eqref{InvTrafo1}.
	\end{itemize}
\end{algo}

The linear mappings ${P}_i,i=1,\dots,n_P,$ in \eqref{QuadraticDecomp} always exist as the left hand side is linear in $\tau_i$.
On the hand, the quadratic decompositions \eqref{QuadraticDecomp} are in general not unique due to the non-unique square matricial representation \cite{SOSDecomp}. Indeed, any polynomial $q({v})$ can be written as $q(v)= m(v)^T(Q+L(\alpha)) m(v)$ where $m(v)$ is a vector of monomials with $q(v)=m(v)^TQm(v)$, and $L(\alpha),\alpha\in\mathbb{R}^{\mu},$ is a linear parametrization of the linear space $\mathcal{L}=\{L=L^T: m^TL(\alpha)m=0\}$. Hence,
\begin{align}\label{ExactSOS}
	z_i\tau_ip_i &=	 \sum_{j=1}^{\beta}\tau_i[j]z^T(Q_{i,j}+L_{i,j}(\alpha_{i,j}))z\notag\\ &=\sum_{j=1}^{\beta}z^T(\tau_i[j]Q_{i,j}+L_{i,j}(\tilde{\alpha}_{i,j}))z
\end{align}
where $\tau_i[j]$ and $z_i[j]$ denote the $j$-th element of $\tau_i$ and $z_i$, respectively, $z^T(Q_{i,j}+L_{i,j}(\alpha_{i,j}))z$ is the square matricial representation of $z_i[j]p_i$, and $\tilde{\alpha}_{i,j}=\tau_i[j]\alpha_{i,j}$. Since the square matricial representation \eqref{ExactSOS} of $z_i\tau_ip_i$ is linear in the optimization variables $\tau_i$ and $\tilde{\alpha}_{i,j}$, we could replace the quadratic decompositions \eqref{QuadraticDecomp} by the square matricial representations \eqref{ExactSOS} in order to deteriorate the conservatism of condition \eqref{LMICond} due to the additional degrees of freedom, i.e., $\tilde{\alpha}_{i,j}$. Note that these additional parameters of the square matricial representation are automatically exploited by SOS solvers as YALMIP \cite{YALMIP}. Therefore, Theorem~\ref{SynAENLM} with the quadratic decompositions \eqref{ExactSOS} instead of \eqref{QuadraticDecomp} incorporates actually the same accuracy as the SOS condition from Lemma~\ref{ThmAnaNLM}.\\\indent 
As a last comment, the SOS condition of \eqref{DissiAENLM} boils down to a matrix inequality independent of $x$, for which LMI techniques from linear robust control can be applied, because we write explicitly the SOS decomposition \eqref{QuadraticDecomp} or \eqref{ExactSOS} and we apply Finsler's lemma to connect signal $x$ and $z$. Both steps are not required in the SOS condition of Lemma~\ref{ThmAnaNLM} as they would be done by an SOS solver.


\begin{rmk}\label{RelaxFinsler}
In Theorem~\ref{SynAENLM}, we include via Finsler's lemma the equality constraint $x-T_xz=0$ which is equivalent to $(x-T_xz)^T(x-T_xz)\leq0$. Since equality constraints might result in numerical problems, we relax this constraint in our implementations by $(x-T_xz)^T(x-T_xz)\leq z^TQ_xz$ for some $Q_x\succeq0$, which can be included into \eqref{LMICond} by the S-procedure.
\end{rmk}


\subsection{Numerical calculation of AE-NLM}\label{SecNumNLM}

We obtain from data by Theorem~\ref{SynAENLM} an upper bound on the AE-NLM and the corresponding optimal linear surrogate model for the system
\begin{align}\label{NumEx2}
	\begin{bmatrix}x_1(t+1)\\x_2(t+1)
	\end{bmatrix}&=\begin{bmatrix}
	0.3x_1+x_2^3\\
	0.2x_2+0.1x_2^2-0.3x_1^3+0.4u
	\end{bmatrix}(t)
\end{align} 
with operation set $x_1^2\leq 1$, $x_2^2\leq 1$, and $u^2\leq1.5^2$ and $y(t)=x(t)$. We suppose the access to samples \eqref{DataSet} from one trajectory with initial condition $x(0)=\begin{bmatrix}-1 & -1\end{bmatrix}^T$, $u(t)=1.5\sin(0.002t^2+0.1t)$, and noise with constant signal-to-noise-ratio $||\tilde{d}_i||_2\leq0.02||\tilde{x}_i||_2$. Moreover, let $z(x,u)=\begin{bmatrix}x_1&x_2& x_2^2&x_1^3&x_2^3&u\end{bmatrix}^T$ be known.\\\indent 
For $\Sigma_F=\Sigma_\text{p}$, we compute from the available data the upper bounds $0.6751\ (S=10), 0.5910\ (S=20)$, and $0.4823\ (S=50)$ for the AE-NLM and the linear approximation \eqref{LinearSystem} with 
\begin{equation}\label{NumLinSys}
\begin{aligned}   
&A_\Psi=\begin{bmatrix}0.2941 &   -0.0751\\	0.2548  &  -0.0442	\end{bmatrix}, \quad B_{\Psi}=\begin{bmatrix}
1.3998\\	-1.1848	\end{bmatrix}, \\    
&C_\Psi=\begin{bmatrix}0.3521  & 0.5150\\	0.3000 &  -0.1060	\end{bmatrix}, \quad D_{\Psi}=10^{-4}\begin{bmatrix}	-0.1229\\
0.0014\end{bmatrix} 
\end{aligned}
\end{equation}
for $S=50$. We also calculate an upper bound $0.3666$ for the AE-NLM using the system dynamics directly by solving an SDP which can be deduced analogously to Theorem~\ref{SynAENLM}.\\\indent
Figure~\ref{Fig.Outputs} shows the state trajectory of system \eqref{NumEx2}, the linear approximation \eqref{LinearSystem} with \eqref{NumLinSys}, and the Jacobian linearization of \eqref{NumEx2} at $x=0$ for the input $u(t)=1.4\sin(0.17t)$. 
\begin{figure}
	\centering
	\includegraphics[width=1\linewidth]{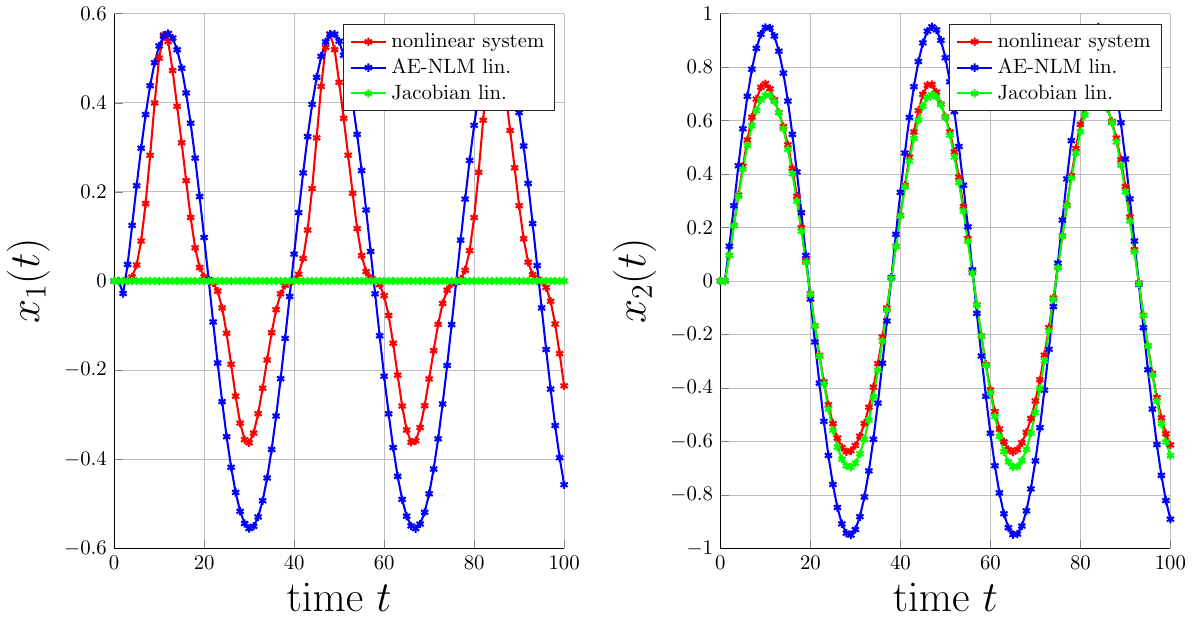}
	\caption{State trajectory of system \eqref{NumEx2}, the linear approximation \eqref{LinearSystem} with \eqref{NumLinSys}, and the Jacobian linearization of \eqref{NumEx2} at $x=0$.}
	\label{Fig.Outputs}
\end{figure}
The figure demonstrates that the Jacobian linearization approximates the $x_2$-dynamics well but fails with respect to the $x_1$-dynamics while the optimized linear approximation \eqref{NumLinSys} yields a balanced approximation of the $x_1$- and $x_2$-dynamics. 
Moreover, the \textquoteleft best\textquoteright\ linear model with $\Phi^{\mathcal{U},\mathcal{G}}_{\text{AE}}=0.4823$ almost halves the worst-case approximation error compared to the Jacobian linearization which yields a data-driven upper bound of $0.9072$ for the AE-NLM by solving \eqref{LMIAna}. Hence, the Jacobian linearization performs in this example barely better than the trivial approximation model with zero matrices in \eqref{LinearSystem} which corresponds to the $\ell_2$-gain of $1.1301$. {Thereby}, we conclude that a robust controller design with our \textquoteleft optimal\textquoteright\ linear model would perform better than with the Jacobian linearization.\\\indent
Furthermore, Figure~\ref{Fig.NLMUpperUpper} shows $\Phi^{\mathcal{U},\mathcal{G}}_{\text{AE}}$ for different sizes of the operation set where for each set a new linear surrogate model is calculated from the same data. 
\begin{figure}
	\centering
	\includegraphics[width=0.9\linewidth]{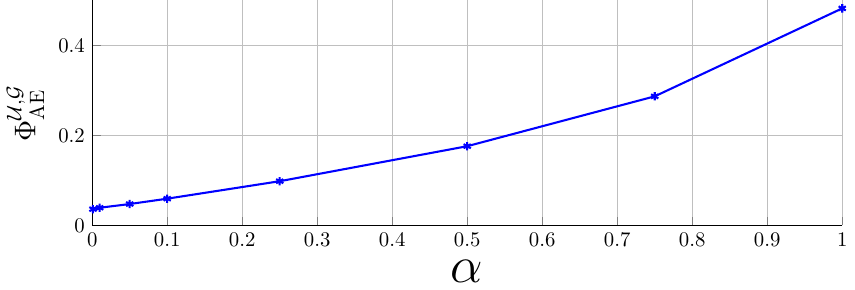}
	\caption{{AE-NLM} for increasing operation set $x_1^2\leq \alpha,$ $x_2^2\leq \alpha$.}
	\label{Fig.NLMUpperUpper}
\end{figure}
Observe that the NLM does not tend to zero for $\alpha\rightarrow0$, even though the nonlinearity vanishes, because the linear part of the system dynamics is still uncertain. Thus, for small $\alpha$, Algorithm~\ref{Algorithm1} fits a linear approximation model for a set of almost linear systems, and therefore the approximation error does not vanish even for small operation sets.

\section{Determining optimal input-output properties}\label{Ext}

The focus of this section is the extension of Theorem~\ref{SynAENLM} to determine more general optimal input-output properties specified by certain classes of time domain hard IQCs while the overall procedure stays as in Algorithm~\ref{Algorithm1}. Contrary to \cite{AnneIQC}, we investigate IQCs over the infinite time horizon, for polynomial systems, and for linear filters parametrized by a general state-space representation.
\begin{coro}[Data-driven inference on IQCs]\label{CoroSynIQC}
Suppose that the data samples \eqref{DataSet} satisfy Assumption~\ref{Noise1} and there exist matrices $T_x\in\mathbb{R}^{n_x\times n_z}$ and $T_u\in\mathbb{R}^{n_u\times n_z}$ with $x=T_x z$ and $u=T_u z$, respectively. If there exist matrices $X,Y^{-1}\succ0$, non-negative scalars $\tau_{\Sigma 1},\dots, \tau_{\Sigma n_S}, \tau_x$, a vector $\gamma\in\mathbb{R}^{n_\gamma}$, matrices $\tilde{K}\in\mathbb{R}^{n_x\times n_x},L\in\mathbb{R}^{n_x\times (n_u+n_y)},\tilde{M}\in\mathbb{R}^{n_{p_2}\times n_x},N\in\mathbb{R}^{n_{p_2}\times (n_u+n_y)}$, and polynomials $z_i\tau_i\in\text{SOS}[x,u],i=1,\dots,n_P$, as described in Theorem~\ref{SynAENLM}, satisfying \eqref{TransCond1} and \eqref{LMICondIQC} with 
\begin{align*}
	\Omega_1 &= \begin{bmatrix}\begin{array}{c|c}\begin{matrix}0 & 0\\ \tilde{K} & 0\end{matrix} & \begin{matrix}0 & I_{n_x}\\ L\begin{bmatrix}
	T_u\\H^*\end{bmatrix}\phantom{\Bigg|} & X \end{matrix}\\\hline \begin{matrix}\tilde{M} & 0\end{matrix} & \begin{matrix} N\begin{bmatrix}
	T_u\\H^*
	\end{bmatrix}\phantom{\Bigg|} & 0	\end{matrix}\end{array}
	\end{bmatrix},\\
	\Omega_2 &= \begin{bmatrix}\hspace{0.2cm}\begin{matrix}	I_{2n_x} & 0 & 0 \\	\hline 
	0 & D_{u1}T_u+D_{y1}H^* & 0\\\hline 0 & I_{n_z} & 0 \\ 0 & 0 & I_{n_x}\\\hline 
	0 & I_{n_z} & 0\\\hline \begin{bmatrix}I_{n_x} & I_{n_x}\end{bmatrix}\phantom{\Big|} & -T_x & 0\end{matrix}\hspace{0.2cm}\end{bmatrix},
	\end{align*}
and
\begin{equation*}
	\Omega_3 = \begin{bmatrix}\begin{array}{c|c|c|c|c}
	0 & 0 &0 & 0 & 0\\\hline 0 & M_2  & 0 &0&0 
	\end{array}	\end{bmatrix},
\end{equation*}
then all trajectories of the ground-truth polynomial system \eqref{NLsys} with {\eqref{NLsysPoly} and} $(x(t),u(t))\in\mathbb{P}, t\in\mathbb{N}_{[0,N]},$ satisfy for all $N\geq0$ the time domain hard IQC
\begin{equation}\label{HardIQC}
	\sum_{t=0}^{N}\begin{bmatrix}p_1(t)\\p_2(t)	\end{bmatrix}^T\begin{bmatrix}M_1(\gamma) & M_2\\ M_2^T & M_3(\gamma)\end{bmatrix}\begin{bmatrix}p_1(t)\\p_2(t)	\end{bmatrix}\geq0
\end{equation}
with $M_3(\gamma)\prec0$ for all $\gamma\in\mathbb{R}^{n_\gamma}$ and $M_1,M_3^{-1}$ linear in $\gamma$ and the linear filter
\begin{align*}
	x_\Psi(t+1)&=A_\Psi x_\Psi(t)+B_uu(t)+B_yy(t), x_\Psi(0)=0\\
	p_1(t)&=D_{u1}u(t)+D_{y1}y(t)\\
	p_2(t)&=C_\Psi x_\Psi(t)+D_{u2}u(t)+D_{y2}y(t)
\end{align*}
with given matrices $D_{u1}$ and $D_{y1}$ and optimized matrices $A_\Psi$, $B_u$, $B_y$, $C_\Psi$, $D_{u2}$, and $D_{y2}$ from
\begin{equation*}
	\begin{bmatrix}K & L\\ M & N\end{bmatrix}{=}\begin{bmatrix}U & 0\\ 0& I_{n_{p_2}}\end{bmatrix}\begin{bmatrix}A_\Psi & \begin{bmatrix}
	B_{u} & B_y\end{bmatrix}\\ C_\Psi & \begin{bmatrix}
	D_{u2} & D_{y2}\end{bmatrix}\end{bmatrix}\begin{bmatrix}V^T & 0\\ 0 & I_{n_{u}+n_y}\end{bmatrix}
\end{equation*}
with $K=\tilde{K}Y$, $M=\tilde{M}Y$, and $I_{n_x}-XY=UV^T$.
\end{coro}
\begin{figure*}
\begin{align}\label{LMICondIQC}
	&0\preceq\begin{bmatrix}
	\Omega_2^T\text{diag}\left(\mathcal{Y}^T\mathcal{X}\mathcal{Y}\bigg|M_1(\gamma)\bigg|\sum_{i=1}^{n_S}\tau_{\Sigma i}\varDelta_{*i}\bigg|
	\sum_{i=1}^{n_P}P_i(\tau_i)\bigg|\tau_x I_{n_x}\right)
	\Omega_2+\Omega_1^T\Omega_3\Omega_2+\Omega_2^T\Omega_3^T\Omega_1
	& \Omega_1^T\\
	\Omega_1 & \text{diag}\left(\mathcal{Y}^T\mathcal{X}\mathcal{Y},M_3(\gamma)^{-1}\right)	\end{bmatrix}
\end{align}
	\begin{tikzpicture}
	\draw[-,line width=0.7pt] (0,0) -- (18,0);
	\end{tikzpicture}
\end{figure*}
\begin{proof}
The claim follows analogously to Theorem~\ref{SynAENLM}. Applying the Schur complement on \eqref{LMICondIQC} as in \cite{SchererLMI} (Lemma 4.2), then using the congruence transformation from Theorem~\ref{SynAENLM} including $\mathcal{Y}$, and thereafter exploiting the generalized S-procedure from Lemma~\ref{ThmAnaNLM} yield that \eqref{TransCond1} and \eqref{LMICondIQC} imply
\begin{align*}
	0\leq& \begin{bmatrix}p_1\\p_2	\end{bmatrix}^T\begin{bmatrix}M_1(\gamma) & M_2\\ M_2^T & M_3(\gamma)\end{bmatrix}\begin{bmatrix}p_1\\p_2	\end{bmatrix}+\begin{bmatrix}x\\x_\Psi\end{bmatrix}^T\mathcal{X}\begin{bmatrix}x\\x_\Psi\end{bmatrix}\\
	&-\begin{bmatrix}Fz\\A_\Psi x_\Psi+B_u u+B_y y\end{bmatrix}^T\mathcal{X}\begin{bmatrix}Fz\\A_\Psi x_\Psi+B_u u+B_y y\end{bmatrix}
\end{align*}
for all $(x,u,x_\Psi)\in\mathbb{P}\times\mathbb{R}^{n_x}$ and $F\in\Sigma_F$. Since $F^*\in\Sigma_F$, all trajectories of the ground-truth polynomial system \eqref{NLsys} with $(x(t),u(t))\in\mathbb{P}, t\in\mathbb{N}_{[0,N]},$ satisfy for all $N\geq0$
\begin{align*}
	0\leq& \sum_{t=0}^{N}\begin{bmatrix}p_1(t)\\p_2(t)	\end{bmatrix}^T\begin{bmatrix}M_1(\gamma) & M_2\\ M_2^T & M_3(\gamma)\end{bmatrix}\begin{bmatrix}p_1(t)\\p_2(t)	\end{bmatrix}\\
	&+\begin{bmatrix}x(0)\\x_\Psi(0)\end{bmatrix}^T\mathcal{X}\begin{bmatrix}x(0)\\x_\Psi(0)\end{bmatrix}-\star^T\mathcal{X}\begin{bmatrix}x(N+1)\\x_\Psi(N+1)\end{bmatrix}\\
	\leq& \sum_{t=0}^{N}\begin{bmatrix}p_1(t)\\p_2(t)	\end{bmatrix}^T\begin{bmatrix}M_1(\gamma) & M_2\\ M_2^T & M_3(\gamma)\end{bmatrix}\begin{bmatrix}p_1(t)\\p_2(t)	\end{bmatrix}
\end{align*}
by $x(0)=x_\Psi(0)=0$ and $\mathcal{X}\succ0$.
\end{proof}

Since condition \eqref{LMICondIQC} depends linearly on $\gamma$, we can determine the tightest IQC by minimizing over $c^T\gamma$ for a given weighting vector $c\in\mathbb{R}^{n_\gamma}$. Moreover, Corollary~\ref{CoroSynIQC} includes Theorem~\ref{SynAENLM} as special case by the discussion in Section~\ref{NLMIntro}.

\subsection{Further investigation of NLMs}

In Section~\ref{SecGeneralSOS}, we focused on the examination of the AE-NLM from Definition~\ref{NLMDef}. However, \cite{TS} proposes further NLMs based on distinct interconnections of the nonlinear system $H$ and the linear approximation $G$.

\begin{defn}[Further NLMs]\label{DefNLMeasures}
The nonlinearity of a causal stable nonlinear system $H:\mathcal{U}\rightarrow\mathcal{Y}$ is measured by
\begin{align*}
	\Phi^{\mathcal{U},\mathcal{G}}_{\text{IMOE}}&=\inf_{G\in\mathcal{G}}\sup_{\substack{u\in\mathcal{U}:H(u)\neq0\\T\in\mathbb{N}_0}}\frac{||H(u)_{T}-G(u)_T||_{\ell_2}}{||H(u)_T||_{\ell_2}},\\
	\Phi^{\mathcal{U},\mathcal{G}^{inv}}_{\text{MIE}}&=\inf_{G^{-1}\in\mathcal{G}^{inv}}\sup_{\substack{u\in\mathcal{U}\backslash \{0\}\\T\in\mathbb{N}_0}}\frac{||G^{-1}(H(u))_T-u_T||_{\ell_2}}{||u_T||_{\ell_2}},\\
	\Phi^{\mathcal{U},\mathcal{G}^{inv}}_{\text{FE}}&=\inf_{G^{-1}\in\mathcal{G}^{inv}}\sup_{\substack{u\in\mathcal{U}:H(u)\neq0\\T\in\mathbb{N}_0}}\frac{||G^{-1}(H(u))_T-u_T||_{\ell_2}}{||H(u)_T||_{\ell_2}},
\end{align*}
where $G:\mathcal{U}\rightarrow\mathcal{\ell}_{2e}^{n_y}$ and $G^{-1}:\mathcal{Y}\rightarrow\mathcal{U}$ are elements of sets $\mathcal{G}$ and $\mathcal{G}^{inv}$, respectively, of stable linear systems.
\end{defn}

For the existence and well-definedness of these NLMs, we refer to \cite{TS}. In contrast to AE-NLM, the inverse multiplicative output error NLM (IMOE-NLM) and the multiplicative input error NLM (MIE-NLM) are normalized, i.e., a NLM close to one indicates a strong nonlinear input-output behaviour. Intuitively, the IMOE-NLM corresponds to the output-to-error-ratio for the worst case input. To conclude on IMOE-NLM, we apply Corollary~\ref{CoroSynIQC} with $B_y=0$ and $D_{y2}=-I_{n_y}$, which can be imposed by $L=\begin{bmatrix}
\tilde{L} & 0\end{bmatrix}$ and $N=\begin{bmatrix}\tilde{N} & -I_{n_y}\end{bmatrix}$, and $D_{u1}=0$, $D_{y1}=I_{n_y}$, $M_1=\gamma I_{n_y}$, $M_2=0$, $M_3=-\frac{1}{\gamma}I_{n_y}$ which corresponds to the dissipativity of the interconnection in Figure~\ref{Fig:NLMDef} with respect to the supply rate $s(y,e)=\gamma ||y||_2^2-\frac{1}{\gamma}||e||_2^2$. Then the minimal $\gamma$ corresponds to the minimal upper bound on the IMOE-NLM.\\\indent
For the MIE-NLM and the feedback error NLM (FE-NLM), the inverse of the input-output behaviour of the nonlinear system is approximated. To infer on MIE-NLM, consider the interconnection in Figure~\ref{Fig:NLMMIEDef}. 
\begin{figure}
	\begin{center}
		\begin{tikzpicture}[scale=0.5]

\path (2,0) node[rectangle,draw,minimum width=12mm, minimum height=8mm](G) {$H$};
\path (G)+(4,0) node[rectangle,draw,minimum width=12mm, minimum height=8mm](Delta) {$G^{-1}$};

\path (Delta)+(3,0) node[circle,draw,inner sep=0pt, minimum size=3mm](rSum){};
\path (rSum)+(0.5,0.5) node(minus){};
\draw[->,line width=1.5pt] (G)+(-3,0) -- (G) node[pos=0.1, above=3pt,inner sep=0pt](HP2){$u$};
\draw[->,line width=1.5pt] (G) -- (Delta);
\draw[->,line width=1.5pt] (Delta) -- (rSum) node[pos=0.6, below=1.5pt,inner sep=0pt](HP2){${-}$};
\draw[-,line width=1.5pt] (G)+(-2,0) |- (9.05,1.5);
\draw[->,line width=1.5pt] (9,1.5) -- (rSum);
\draw[->,line width=1.5pt] (rSum) -- +(2.3,0) node[pos=0.8, above=3pt,inner sep=0pt](HP2){$e$};
    	
\end{tikzpicture}
	\end{center}
	\caption{Interconnection of $H$ and $G^{-1}$ for MIE-NLM and FE-NLM.}
	\label{Fig:NLMMIEDef}
\end{figure}
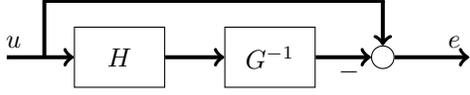
Thus, Corollary~\ref{CoroSynIQC} can be employed with $B_u=0$ and $D_{u2}=-I_{n_u}$, which can be imposed by $L=\begin{bmatrix}0 & \tilde{L}\end{bmatrix}$ and $N=\begin{bmatrix}-I_{n_u} & \tilde{N}\end{bmatrix}$, and $D_{y1}=0$, $D_{u1}=I_{n_u}$, $M_1=\gamma I_{n_u}$, $M_2=0$, $M_3=-\frac{1}{\gamma}I_{n_y}$. Furthermore, we gather an upper bound on FE-NLM by Corollary~\ref{CoroSynIQC} by the same setup as for MIE-NLM but $D_{y1}=I_{n_y}$ and $D_{u1}=0$.

\begin{rmk}[Linear filter design]
A related problem to determine NLMs is the linear filter design with performance guarantees \cite{FilterDesign}. To this end, consider the polynomial system
\begin{align*}
x(t+1)&=F^*z(x(t),u(t)),x(0)=0\\
y(t) &= H_y z(x(t),u(t))\\
p(t) &= H_p z(x(t),u(t))
\end{align*}
with unknown coefficients $F^*$, measured signal $y$ with known $H_y$, and the to-be-estimated signal $p$ with known $H_p$. We aim to design the linear filter
\begin{align*}
	x_\Psi(t+1)&=A_\Psi x_\Psi(t)+ B_uu(t)+ B_y y(t),x_\Psi(0)=0\\
	p_\Psi(t) &= C_\Psi x_\Psi(t)+D_uu(t)+D_yy(t)
\end{align*}
such that the $\ell_2$-gain from $u$ to $e=p-p_\Psi$ is minimal for all trajectories of the polynomial with $(x(t),u(t))\in\mathbb{P},\forall t\geq0$. The sufficient LMI conditions follow directly from Corollary~\ref{CoroSynIQC} with $D_{u1}=I_{n_u}, D_{y1}=0, M_1(\gamma)=\gamma I_{n_u}$, $M_2=0$, $M_3(\gamma)=-\frac{1}{\gamma}I_{n_p}$, $p_\Psi$ corresponds to $p_2$, and minimizing over $\gamma$. Moreover, we must modify $\Omega_1$ to
\begin{equation*}
	\begin{bmatrix}\begin{array}{c|c}\begin{matrix}0 & 0\\ \tilde{K} & 0\end{matrix} & \begin{matrix}0 & I_{n_x}\\ L\begin{bmatrix}
	T_u\\H^*\end{bmatrix}\phantom{\Bigg|} & X \end{matrix}\\\hline \begin{matrix}-\tilde{M} & 0\end{matrix} & \begin{matrix} H_p-N\begin{bmatrix}
	T_u\\H^*
	\end{bmatrix}\phantom{\Bigg|} & 0	\end{matrix}\end{array}
	\end{bmatrix}
\end{equation*}
as we require the sum quadratic constraint
\begin{equation*}
	\sum_{t=0}^{N}\begin{bmatrix}u(t)\\p-p_\Psi(t)\end{bmatrix}^T\begin{bmatrix}\gamma I_{n_u} & 0\\ 0 & -\frac{1}{\gamma}I_{n_p}\end{bmatrix}\begin{bmatrix}u(t)\\p-p_\Psi(t)\end{bmatrix}\geq0,
\end{equation*}
instead of the hard IQC \eqref{HardIQC}.
\end{rmk}

\begin{rmk}[Continuous-time system]
The presented results Lemma~\ref{ThmAnaNLM}, Theorem~\ref{SynAENLM}, and Corollary~\ref{CoroSynIQC} can easily be formulated for continuous-time polynomial systems following \cite{SchererLMI}. 
\end{rmk}

\begin{rmk}[NLM for unstable systems]\label{LinSysDyn}

The input-output behaviour of a nonlinear system $H$ with unbounded $\ell_2$-gain renders the NLMs of Definition~\ref{NLMDef} and \ref{DefNLMeasures} to be unbounded which rises the question how the nonlinearity of unstable (polynomial) systems can be measured? To this end, we consider the linear system $x_L(t+1)=A x_L(t)+B u(t)$ that minimizes regarding the unidentified polynomial system $x(t+1)=F^*z(x(t),u(t))$ the normalized Euclidean-norm of the error $e(x,u)=F^*z(x,u)-Ax-Bu$ within \eqref{Constraints}, i.e., 
\begin{align*}
	&\gamma^*=\min_{\gamma\geq0, (A,B)\in\mathbb{R}^{n_x\times n_x}\times\mathbb{R}^{n_x\times n_u}} \gamma\\
	&\text{s.t.\ } ||F^*z(x,u)-Ax-Bu||_2^2\leq\gamma^2\ \Bigg|\Bigg|\begin{bmatrix}x\\u\end{bmatrix}\Bigg|\Bigg|_2^2,\ \forall (x,u)\in\mathbb{P}.
\end{align*}
The obtained linearization corresponds to the Jacobian linearization of $x(t+1)=F^*z(x(t),u(t))$ at $\begin{bmatrix}x^T & u^T\end{bmatrix}^T=0$ if the operation set $\mathbb{P}$ tends to $\{(0,0)\}$. Exploiting the set-membership $F^*\in\Sigma_F$ and polynomials $z_i\tau_i\in\text{SOS}[x,u],i=1,\dots,n_P$, as in Theorem~\ref{SynAENLM}, we derive a data-based upper bound of $\gamma^*$
\begin{align*}
	\gamma^*\leq&\min_{\substack{\gamma\geq0, (A,B)\in\mathbb{R}^{n_x\times n_x}\times\mathbb{R}^{n_x\times n_u},\\ \tau_{\Sigma 1},\dots, \tau_{\Sigma n_S}\geq0,z_1\tau_1,\dots,z_{n_P}\tau_{n_P}\in\text{SOS}[x,u]}} \gamma\\
	&\text{s.t.\ } 0\preceq
	\star^T\Theta\begin{bmatrix}\hspace{0.2cm}\begin{matrix}	 0 & \begin{bmatrix}T_x \\  T_u\end{bmatrix}\vspace{0.1cm}\\	\hline 
	I_{n_x} & -AT_x-BT_u\\\hline 0 & I_{n_z}  \\  I_{n_x} & 0\\\hline 
	0 & I_{n_z} \end{matrix}\hspace{0.2cm}\end{bmatrix}
\end{align*}
with $x=T_x z$, $u=T_u z$, and
\begin{align*}
	\Theta=\text{diag}\left(\gamma^2I_{n_x+n_u}\bigg|-I_{n_x}\bigg|\sum_{i=1}^{n_S}\tau_{\Sigma i}\varDelta_{*i}\bigg|\sum_{i=1}^{n_P}P_i(\tau_i)	\right).
\end{align*}
Note that the Schur complement renders this optimization problem linear regarding the variables $A$ and $B$, and hence yields an SDP. 
\end{rmk}

\section{Further investigation of set-memberships for coefficient matrices}\label{DataPre2}

Since the computation of $\Sigma_\text{p}$ calls for the solution of LMI \eqref{Elli_LMI} which might be computationally expensive for a large number of samples, we present in Section~\ref{Supersets2} two further supersets of $\Sigma$. Moreover, we show in Section~\ref{Asymptotic} that all three supersets converge to $F^*$ despite noisy data if the number of samples tends to infinity and if further assumptions hold. Finally, we compare the accuracy of the supersets for determining the $\ell_2$-gain in a numerical example in Section~\ref{NumComparison}.

\subsection{Supersets for $\mathit{\Sigma}$}\label{Supersets2}

Whereas the computation of $\Sigma_\text{p}$ requires to solve an SDP which complexity increases linearly with the number of samples, another superset was suggested in \cite{Groningen} which can be derived without additional optimization. To this end, we reformulate the pointwise noise bound from Assumption~\ref{Noise1} to a characterization as in \cite{Groningen} where the noise realizations $\tilde{d}_1^T,\dots,\tilde{d}_S^T$ are bounded cumulatively.
\begin{lem}[Cumulatively bounded noise]\label{NoiseAv}
The matrix of noise realizations $\tilde{D}=\begin{bmatrix}\tilde{d}_1&\cdots&\tilde{d}_S\end{bmatrix}$ with $\tilde{d}_i\in\mathcal{D}_i$ from Assumption~\ref{Noise1} is an element of
\begin{equation}\label{noise_Av}
	\mathcal{D}_\text{c}{=}\left\{D\in\mathbb{R}^{n_x\times S}{:}\begin{bmatrix}D^T\\I_{n_x}\end{bmatrix}^T\begin{bmatrix}\tilde{\Delta}_1 & \tilde{\Delta}_2\\ \tilde{\Delta}_2^T & \tilde{\Delta}_3\end{bmatrix}\begin{bmatrix}D^T\\I_{n_x}\end{bmatrix}\preceq0 \right\}
\end{equation}
with $\tilde{\Delta}_1=-\text{diag}(\Delta_{1,i},\dots,\Delta_{1,S})\succ0$, $\tilde{\Delta}_2=\begin{bmatrix}\Delta_{2,1}^T&\cdots&\Delta_{2,S}^T	\end{bmatrix}^T$, and $\tilde{\Delta}_3=-\sum_{i=1}^{S}\Delta_{3,i}$.	
\end{lem}
\begin{proof}
The proof of Lemma~\ref{DualSigma} already shows that the noise bounds from \eqref{SepPrimnoise} are equivalent to the dual version \eqref{DualNoise}. The summation of \eqref{DualNoise} over all noise realizations yields 
\begin{align*}
	\sum_{i=1}^{S}\begin{bmatrix}\tilde{d}_i^T\\I_{n_x}\end{bmatrix}^T\begin{bmatrix}-\Delta_{1,i} & \Delta_{2,i}\\ \Delta_{2,i}^T & -\Delta_{3,i}\end{bmatrix}\begin{bmatrix}\tilde{d}_i^T\\I_{n_x}\end{bmatrix}\preceq0,
\end{align*}
which is equivalent to \eqref{noise_Av} by simple reformulation.
\end{proof}

Due to the summation of \eqref{DualNoise} over all noise realizations, $\mathcal{D}_{\text{c}}$ facilitates more noise realizations than the original pointwise descriptions $\mathcal{D}_i$. Thus, the non-tight characterization \eqref{noise_Av} amounts to a non-tight set-membership representation of $\Sigma$ in the following proposition. 
\begin{lrop}[Cumulative superset of $\mathit{\Sigma}$]\label{AvSigma}
Suppose $\tilde{Z}=\begin{bmatrix}z(\tilde{x}_1,\tilde{u}_1)&\cdots&z(\tilde{x}_S,\tilde{u}_S)\end{bmatrix}$ is full row rank and the inverse of 
\begin{align*}
	&\Delta_\text{c}=\begin{bmatrix}\Delta_{1\text{c}} & \Delta_{2\text{c}}\\ \Delta_{2\text{c}}^T & \Delta_{3\text{c}}\end{bmatrix}\\
	&\hspace{0.2cm} {=}\begin{bmatrix}\tilde{Z}\tilde{\Delta}_1\tilde{Z}^T & -\tilde{Z}(\tilde{\Delta}_1\tilde{X}^{+^T}+\tilde{\Delta}_2)\\-(\tilde{X}^+\tilde{\Delta}_1^T{+}\tilde{\Delta}_2^T)\tilde{Z}^T & \begin{bmatrix}\tilde{X}^{+^T}\\I_{n_x}\end{bmatrix}^T\begin{bmatrix}\tilde{\Delta}_1 & \tilde{\Delta}_2\\ \tilde{\Delta}_2^T & \tilde{\Delta}_3\end{bmatrix}\begin{bmatrix}\tilde{X}^{+^T}\\I_{n_x}\end{bmatrix} \end{bmatrix}
\end{align*}
exists for the data-dependent matrix $\tilde{X}^+=\begin{bmatrix}\tilde{x}_1^+&\cdots&\tilde{x}_S^+\end{bmatrix}$. Then the set of feasible coefficients $\Sigma$ is a subset of 
\begin{align}\label{Sigma_av}
	\Sigma_{\text{c}}=\left\{F\in\mathbb{R}^{n_x\times n_z}:\begin{bmatrix}I_{n_z}\\F\end{bmatrix}^T\varDelta_{\text{c}}\begin{bmatrix}I_{n_z}\\F\end{bmatrix}\preceq0\right\}
\end{align}
with $\varDelta_{\text{c}}=\begin{bmatrix}-\varDelta_{1\text{c}} & \varDelta_{2\text{c}}\\ \varDelta_{2\text{c}}^T & -\varDelta_{3\text{c}}\end{bmatrix}$ and $\begin{bmatrix}\varDelta_{1\text{c}} & \varDelta_{2\text{c}}\\ \varDelta_{2\text{c}}^T & \varDelta_{3\text{c}}\end{bmatrix}=\Delta_\text{c}^{-1}$.	
\end{lrop}\vspace{0.1cm}
\begin{proof}
Analogously to \cite{Groningen} (Lemma 4), combining \eqref{noise_Av}, data samples~\eqref{DataSet}, and the system dynamics $x(t+1)=Fz(x(t), u(t))$ yields the tight description 
\begin{align}\label{FSSA1}
	\left\{F\in\mathbb{R}^{n_x\times n_z}:\begin{bmatrix}F^T\\I_{n_x}\end{bmatrix}^T\Delta_\text{c} \begin{bmatrix}F^T\\I_{n_x}\end{bmatrix}\preceq0\right\}
\end{align}
of the set of feasible coefficients which explain the data \eqref{DataSet} for average noise description \eqref{noise_Av}. Since $\tilde{\Delta}_1\succ0$ according to Lemma~\ref{NoiseAv} and $\tilde{Z}$ is full row rank, $\Delta_{1\text{c}}\succ0$. Together with 
the existence of the inverse of $\Delta_\text{c}$, the dualization lemma can be employed on \eqref{FSSA1} to derive the equivalent description \eqref{Sigma_av}.
\end{proof}

As already indicated in \cite{MartinDissi}, the summation in Lemma~\ref{NoiseAv} corresponds to the S-procedure in Proposition~\ref{Primal_elli} for $\alpha_1=\cdots=\alpha_S=1$, and hence $\Sigma_\text{c}$ is a superset of $\Sigma_\text{p}$. We already discussed the assumption on $\tilde{Z}$ after Proposition~\ref{Primal_elli}. To the best knowledge of the authors, the invertibility assumption on $\Delta_\text{c}$ can not be dropped in general. However, if this assumption is not satisfied and the number of samples does not allow to solve the LMI~\eqref{Elli_LMI}, we suggest to consider, instead of $\Sigma_\text{c}$, the superset $\Sigma_\text{p}$ together with the feasible solution provided in the proof of Proposition~\ref{Primal_elli}. Since this only requires the full rank of $\tilde{Z}$, it constitutes an alternative to \cite{Groningen} with less assumptions for our purposes.\\\indent
While the supersets $\Sigma_\text{p}$ and $\Sigma_\text{c}$ use one matrix inequality to characterize the feasible coefficients, i.e., $n_S=1$ in \eqref{FSSgen}, we propose next a third superset, inspired by the window noise description in \cite{SetMem_Con}. To this end, we define for a window length $L\leq S$ and $i=1,\dots,S_0= S-L+1,$ the data-dependent matrices $\tilde{X}^+_i=\begin{bmatrix}\tilde{x}_i^+&\cdots&\tilde{x}_{i+L-1}^+\end{bmatrix}$ and $\tilde{Z}_i=\begin{bmatrix}z(\tilde{x}_i,\tilde{u}_i)&\cdots&z(\tilde{x}_{i+L-1},\tilde{u}_{i+L-1})\end{bmatrix}$ and the corresponding noise realizations $\tilde{D}_i=\begin{bmatrix}\tilde{d}_i & \cdots & \tilde{d}_{i+L-1}\end{bmatrix}$ where each satisfies
\begin{equation}\label{DataWindow}
	\begin{bmatrix}\tilde{D}_i^T\\I_{n_x}\end{bmatrix}^T\begin{bmatrix} \tilde{\Delta}_{1,i}& \tilde{\Delta}_{2,i}\\ \tilde{\Delta}_{2,i}^T &\tilde{\Delta}_{3,i} \end{bmatrix}\begin{bmatrix}\tilde{D}_i^T\\I_{n_x}\end{bmatrix}\preceq0
\end{equation}
with $\tilde{\Delta}_{1,i}=-\text{diag}(\Delta_{1,i},\dots,\Delta_{1,i+L-1})\succ0$, $\tilde{\Delta}_{2,i}=\begin{bmatrix}\Delta_{2,i}^T& \cdots & \Delta_{2,i+L-1}^T	\end{bmatrix}^T$, and $\tilde{\Delta}_{3,i}=-\sum_{j=i}^{i+L-1}\Delta_{3,j}$ by Lemma~\ref{NoiseAv}. Thus, we study contrary to \cite{SetMem_Con} general quadratic noise descriptions and unknown coefficient matrices.

\begin{lrop}[Window-based superset of $\mathit\Sigma$]\label{WSigma}
For $i=1,\dots,S_0$, suppose $\tilde{Z}_i$ is full row rank and the inverse of 
\begin{align*}
	&{\Delta}_{\text{w},i}=\begin{bmatrix}{\Delta}_{1\text{w},i} & {\Delta}_{2\text{w},i}\\ {\Delta}_{2\text{w},i}^T & {\Delta}_{3\text{w},i}\end{bmatrix}\\
	& {=}\begin{bmatrix}\tilde{Z}_i\tilde{\Delta}_{1,i}\tilde{Z}_i^T & -\tilde{Z}_i(\tilde{\Delta}_{1,i}\tilde{X}_i^{+^T}+\tilde{\Delta}_{2,i})\\{-}(\tilde{X}_i^+\tilde{\Delta}_{1,i}^T{+}\tilde{\Delta}_{2,i}^T)\tilde{Z}_i^T & \begin{bmatrix}\tilde{X}_i^{+^T}\\I_{n_x}\end{bmatrix}^T\begin{bmatrix}\tilde{\Delta}_{1,i} & \tilde{\Delta}_{2,i}\\ \tilde{\Delta}_{2,i}^T & \tilde{\Delta}_{3,i}\end{bmatrix}\begin{bmatrix}\tilde{X}_i^{+^T}\\I_{n_x}\end{bmatrix} \end{bmatrix}
\end{align*}
exists. Then the set of feasible coefficients $\Sigma$ is a subset of 
\begin{align*}
	\Sigma_{\text{w}}=\left\{F:\begin{bmatrix}I_{n_z}\\F\end{bmatrix}^T\varDelta_{\text{w},i}\begin{bmatrix}I_{n_z}\\F\end{bmatrix}\preceq0,i=1,\dots,S_0\right\}
\end{align*}
with $\varDelta_{\text{w},i}=\begin{bmatrix}-\varDelta_{1\text{w},i} & \varDelta_{2\text{w},i}\\ \varDelta_{2\text{w},i}^T & -\varDelta_{3\text{w},i}\end{bmatrix}$ and $\begin{bmatrix}\varDelta_{1\text{w},i} & \varDelta_{2\text{w},i}\\ \varDelta_{2\text{w},i}^T & \varDelta_{3\text{w},i}\end{bmatrix}={\Delta}_{\text{w},i}^{-1}$.	
\end{lrop}	
\begin{proof}
	The result follows immediately by Proposition~\ref{AvSigma} for each window $i=1,\dots,S_0$.
\end{proof}

Clearly, $\Sigma_\text{w}$ corresponds to $\Sigma_\text{c}$ for $L=S$, and hence $\Sigma_\text{w}\subseteq\Sigma_\text{c}$. Note that the window length $L$ can not be chosen arbitrarily small as otherwise the invertibility of $\Delta_{\text{w},i}$ is violated. To refine the accuracy of $\Sigma_\text{w}$, we could compute an ellipsoidal outer approximation for each window as for $\Sigma_\text{p}$. Thereby, the invertibility of $\Delta_{\text{w},i}$ can be dropped and we meet the pointwise bound from Assumption~\ref{Noise1} tighter than $\Sigma_\text{p}$, due to the additional split of data \eqref{DataSet} into $S_0$ windows.\\\indent
In the context of deriving quadratic matrix inequalities \eqref{FSSgen} for coefficient matrices $F^*$ from noisy samples, we also refer to \cite{ProbDetNoise} if the disturbance is Gaussian distributed.

\subsection{Asymptotic consistency of $\mathit\Sigma_\text{p}, \mathit\Sigma_\text{c}$, and $\mathit\Sigma_\text{w}$}\label{Asymptotic}

We show that $\Sigma_\text{p}, \Sigma_\text{c}$, and $\Sigma_\text{w}$ converge to the true coefficients $F^*$ for infinitely many samples together with a tight noise bound. Furthermore, we derive supersets of $\Sigma_\text{c}$ and $\Sigma_\text{w}$ for non-tight noise bounds and $S\rightarrow\infty$.\\\indent
First, we deduce an auxiliary result to conclude on a set of coefficient matrices which can be falsified by infinitely many samples even if the noise description is not tight. This result can then be applied to evaluate the asymptotic exactness of $\Sigma_\text{p}$ and $\Sigma_\text{w}$. For the data sample \eqref{DataSet}, an $L_0\in\mathbb{N}_{[1,S]}$, and any $t\in\mathbb{N}_{[1,S-L_0+1]}$, we define the matrices
\begin{align*}
	X_t&=\begin{bmatrix}\tilde{x}_t &\cdots & \tilde{x}_{t+L_0-1}\end{bmatrix},\\
	Z_t&=\begin{bmatrix}z(\tilde{x}_t,\tilde{u}_t) &\cdots & z(\tilde{x}_{t+L_0-1},\tilde{u}_{t+L_0-1})\end{bmatrix},\\
	D_t&=\begin{bmatrix}\tilde{d}_t &\cdots & \tilde{d}_{t+L_0-1}\end{bmatrix}
\end{align*} 
with $X_{t+1}=F^{*}Z_{t}+D_{t}$. We suppose the knowledge on a compact set $\mathcal{D}_\text{nt}\subset\mathbb{R}^{n_x\times L_0}$ which contains the noise realizations $D_t$ for all $t\in\mathbb{N}_{[1,S-L_0+1]}$. Since $\mathcal{D}_\text{nt}$ might be a non-tight bound on $D_t$, we assume analogously to \cite{SetMem_MPC} (Assumption 5) that there exists an unknown tight noise bound.
\begin{assu}[Tight noise bound]\label{AsstightBound}
Suppose there exist a compact set $\Omega\subset\mathbb{R}^{n_x\times L_0}$ and $\rho>0$ such that $\Omega\oplus\rho\mathcal{B}\supseteq\mathcal{D}_\text{nt}\supseteq\Omega$ with the unit ball $\mathcal{B}=\{D\in\mathbb{R}^{n_x\times L_0}:||D||_\text{Fr}\leq1\}$. Moreover, for all $t\in\mathbb{N}_{[1,S-L_0+1]}$, let $D_t\in\Omega$ and let a function $p:\mathbb{R}_{>0}\rightarrow \mathbb{R}_{(0,1]}$ exist with $\text{Pr}(||D_t-\bar{D}||_\text{Fr}<\epsilon)\geq p(\epsilon)$ for all $\bar{D}\in\partial\Omega$ and all $\epsilon>0$.
\end{assu}

Assumption~\ref{AsstightBound} supposes that any noise realization matrix, arbitrarily close to the boundary of $\Omega$, can be observed at any time window with non-zero probability, and hence $\Omega$ is a tight noise characterization. Note that Assumption~\ref{AsstightBound} implies that the noise realizations $\tilde{d}_1,\dots,\tilde{d}_S$ are random variables.

\begin{assu}[Conditionally independent disturbance]\label{AssIndepend}
For any $i,j\in\mathbb{N}_{[1,S]},i\neq j$, the disturbance realizations $\tilde{d}_i$ and $\tilde{d}_j$ are conditionally independent. 
\end{assu}

\begin{assu}[Persistent excitation]\label{AssPE}
Suppose there exist positive scalars $\alpha, \beta$, and $L_\text{pe}\leq L_0$ such that $||Z_t||_\text{Fr}\leq\alpha$ for all $t\in\mathbb{N}_{[1,S{-L_0+1}]}$ and 
\begin{equation*}
	\sum_{i=t}^{t+L_\text{pe}-1} z(\tilde{x}_i,\tilde{u}_i)z(\tilde{x}_i,\tilde{u}_i)^T \succeq \beta I_{n_z}
\end{equation*}
for all $t\in\mathbb{N}_{[1,S-L_\text{pe}+1]}$.
\end{assu}

\begin{lem}[Set of falsified coefficients]\label{SetofunFaCoe}
Suppose Assumption~\ref{AsstightBound}, {\ref{AssIndepend},} and \ref{AssPE} hold and a non-tight noise bound $\mathcal{D}_\text{nt}$ is known. Then the coefficients $F\in\mathbb{R}^{n_x\times n_z}$, excluded by
\begin{equation}
	\{F^*\}\oplus\rho\sqrt{\frac{L_\text{pe}}{\beta}}\mathcal{B},
\end{equation}
can be falsified with probability $1$ by the data \eqref{DataSet} for $S\rightarrow\infty$.
\end{lem}  
\begin{proof}
While we follow the main steps from \cite{SetMem_MPC}, we consider contrary to \cite{SetMem_MPC} an ellipsoidal noise description, unknown coefficient matrices, and conclusions from the data matrices $X_t$ and $Z_t$ over $L_0$ time steps.\\
We define the set of coefficients $F$ which are admissible for the data $X_{t+1},Z_{t}$, and $D_{t}\in\mathcal{D}_\text{nt}$
\begin{align*}
	F_t&=\{F\in\mathbb{R}^{n_x\times n_z}: X_{t+1}-FZ_t\in\mathcal{D}_\text{nt}\}\\
	&=\{F\in\mathbb{R}^{n_x\times n_z}: (F^*-F)Z_t+D_{t}\in\mathcal{D}_\text{nt}\}.
\end{align*}
Moreover, we define the matrix normal cone $\mathcal{N}_{\mathcal{D}_\text{nt}}(\hat{D})$ of $\mathcal{D}_\text{nt}$ at the matrix $\hat{D}\in\partial\mathcal{D}_\text{nt}$ as
\begin{equation*}
	\mathcal{N}_{\mathcal{D}_\text{nt}}(\hat{D})=\{G\in\mathbb{R}^{n_x \times L_0}:\left<G,D-\hat{D}\right>_\text{Fr}\leq0, \forall D\in\mathcal{D}_\text{nt}\}.
\end{equation*} 
In the sequel, we use the fact that there exists for any matrix $D\in\mathbb{R}^{n_x \times L_0}$ a matrix $\hat{D}\in\partial\mathcal{D}_\text{nt}$ such that $D\in\mathcal{N}_{\mathcal{D}_\text{nt}}(\hat{D})$. For the matrix normal cone, this is clear because for any $K\in\mathbb{R}^{n_x \times L_0}$ the solution of $\sup_{D\in\mathcal{D}_\text{nt}}\left<K,D\right>_\text{Fr}$ is attained for some $\hat{D}$ by the Weierstrass theorem. Moreover, $\hat{D}\in\partial\mathcal{D}_\text{nt}$ as otherwise the small perturbation $\epsilon K$ of $\hat{D}$ would lead to a feasible and larger solution. Thus, $\left<K,D\right>_\text{Fr}-\left<K,\hat{D}\right>_\text{Fr}\leq 0$ for all $D\in\mathcal{D}_\text{nt}$, and hence $\bigcup_{\hat{D}\in\partial\mathcal{D}_\text{nt}}\mathcal{N}_{\mathcal{D}_\text{nt}}(\hat{D})=\mathbb{R}^{n_x \times L_0}$.\\	Furthermore, the persistent excitation assumption implies 
\begin{equation*}
	Z_iZ_i^T=\sum_{i=t}^{t+L_0-1} z_iz_i^T \succeq \sum_{i=t}^{t+L_\text{pe}-1} z_iz_i^T \succeq\beta I_{n_z}
\end{equation*}
with $z_i=z(\tilde{x}_i,\tilde{u}_i)$, and thus for any $F$
\begin{equation*}
	\sum_{i=t}^{t+L_\text{pe}-1} (F^*-F)Z_iZ_i^T(F^*-F)^T \succeq \beta (F^*-F)(F^*-F)^T.
\end{equation*}
This leads to 
\begin{equation*}
	\sum_{i=t}^{t+L_\text{pe}-1} ||(F^*-F)Z_i||_\text{Fr}^2 \geq \beta ||(F^*-F)||_\text{Fr}^2,
\end{equation*}
as $A\preceq B$ implies $\text{tr}(A) \leq\text{tr}(B)$ and $\text{tr}(AB)=\text{tr}(BA)$, and hereby there exists a $j\in\mathbb{N}_{[t,t+L_\text{pe}-1]}$ such that $||(F^*-F)Z_j||^2_\text{Fr} \geq \frac{\beta}{L_\text{pe}} ||(F^*-F)||^2_\text{Fr}$.\\
With this preparation, we can now show the claim. Consider any coefficient matrix $F$ such that there exists an $\epsilon>0$ with $||F^*-F||_\text{Fr}\geq\epsilon+\rho\sqrt{\frac{L_\text{pe}}{\beta}}$. Together with Assumption~\ref{AssPE}, there exists a $j\in\mathbb{N}_{[t,t+L_\text{pe}-1]}$ such that 
\begin{equation}\label{StartIn}
	||(F^*-F)Z_j||_\text{Fr} \geq \sqrt{\frac{\beta}{L_\text{pe}}} ||(F^*-F)||_\text{Fr}\geq \epsilon\sqrt{\frac{\beta}{L_\text{pe}}}+\rho.
\end{equation}
Moreover, we can construct a matrix $\hat{D}\in\partial\mathcal{D}_\text{nt}$ with $(F^*-F)Z_j\in\mathcal{N}_{\mathcal{D}_\text{nt}}(\hat{D})$ and a matrix $\bar{D}\in\partial\Omega$ with
\begin{equation}\label{StartStartIn}
	||\hat{D}-\bar{D}||_\text{Fr}\leq\rho
\end{equation}
by Assumption~\ref{AsstightBound}. With the Cauchy-Schwarz inequality, we calculate
\begin{align*}
	&\left<(F^*-F)Z_j, (F^*-F)Z_j+D_j-\hat{D}\right>_\text{Fr}\\
	=\ &||(F^*-F)Z_j||_\text{Fr}^2+\left<(F^*-F)Z_j, D_j-\bar{D}\right>_\text{Fr}\\&+\left<(F^*-F)Z_j, \bar{D}-\hat{D}\right>_\text{Fr}\\
	\geq\ & ||(F^*-F)Z_j||_\text{Fr}^2-||(F^*-F)Z_j||_\text{Fr}\ ||D_j-\bar{D}||_\text{Fr}\\&-||(F^*-F)Z_j||_\text{Fr}\ ||\bar{D}-\hat{D}||_\text{Fr}.
\end{align*} 
If $D_j$ satisfies $||D_j-\bar{D}||_\text{Fr}<\epsilon\sqrt{\frac{\beta}{L_\text{pe}}}$ and together with \eqref{StartStartIn}, then we can write further
\begin{align*}
	&\left<(F^*-F)Z_j, (F^*-F)Z_j+D_j-\hat{D}\right>_\text{Fr}\\
	&{>} ||(F^*-F)Z_j||_\text{Fr}\left(||(F^*-F)Z_j||_\text{Fr}-\epsilon\sqrt{\frac{\beta}{L_\text{pe}}}-\rho\right)\overset{\eqref{StartIn}}{\geq}0.
\end{align*}
Thus, $(F^*-F)Z_j+D_j\notin\mathcal{D}_\text{nt}$ as $(F^*-F)Z_j\in\mathcal{N}_{\mathcal{D}_\text{nt}}(\hat{D})$. Hereby, $F\notin F_{j}$ by the Definition of $F_t$, and therefore the coefficients $F$ are falsified by the data $X_{j+1},Z_j$ for any noise realization $D_j$ with $||D_j-\bar{D}||_\text{Fr}<\epsilon\sqrt{\frac{\beta}{L_\text{pe}}}$. This yields 
\begin{equation}\label{Inequ2}
	\text{Pr}(F\notin F_{j})\geq \text{Pr}\left(||D_j-\bar{D}||_\text{Fr}<\epsilon\sqrt{\frac{\beta}{L_\text{pe}}}\right)\geq p\left(\epsilon\sqrt{\frac{\beta}{L_\text{pe}}}\right)
\end{equation}
by Assumption~\ref{AsstightBound}. \\
By \eqref{Inequ2}, we show that any coefficient matrix $F$ with $||F^*-F||_\text{Fr}\geq\epsilon+\rho\sqrt{\frac{L_\text{pe}}{\beta}}$ can be falsified by a finite set of data with non-vanishing probability. For that reason, it remains to prove that this also holds with probability $1$ for $S\rightarrow\infty$. To this end, let $\mathbb{F}_t=\bigcap_{i=1}^{t}F_i$ with
\begin{align*}
	\text{Pr}(F\in \mathbb{F}_t)\leq\ &\text{Pr}(F\in \bigcap_{i\in\mathbb{N}_{[t-L_\text{pe}+1,t]}}F_i | F\in{\mathbb{F}}_{t-L_0-L_\text{pe}} )\\&\cdot\text{Pr}(F\in{\mathbb{F}}_{t-L_0-L_\text{pe}} ).
\end{align*}
Since the noise realizations $D_i$ for $i\in\mathbb{N}_{[1,t-L_0-L_\text{pe}]}$ and for $i\in\mathbb{N}_{[t-L_\text{pe}+1,t]}$ are conditionally independent by Assumption~\ref{AssIndepend}, \eqref{Inequ2} results in
\begin{align*}
	\text{Pr}(F\in \mathbb{F}_t)\leq& \left(1-p\left(\epsilon\sqrt{\frac{\beta}{L_\text{pe}}}\right)\right)\text{Pr}(F\in{\mathbb{F}}_{t-L_0-L_\text{pe}} )\\
	\leq&\dots\leq \left(1-p\left(\epsilon\sqrt{\frac{\beta}{L_\text{pe}}}\right)\right)^{\lfloor t/(L_0+L_\text{pe})\rfloor}.
\end{align*}
Thus, $\sum_{t=1}^{\infty}\text{Pr}(F\in \mathbb{F}_t)$ is finite. The Borel-Cantelli lemma yields $\text{Pr}(F\in \bigcap_{t=1}^{\infty}\bigcup_{k\geq t}^{\infty}\mathbb{F}_k)=\text{Pr}(F\in \bigcap_{t=1}^{\infty}\mathbb{F}_k)=0$ for any $F$ with $||F^*-F||_\text{Fr}\geq\epsilon+\rho\sqrt{\frac{L_\text{pe}}{\beta}}$ and any $\epsilon>0$. Thus, any such kind of coefficient $F$ can be falsified with probability one for infinitely many data points.
\end{proof}

With this auxiliary result, we analyze now the asymptotic accuracy of $\Sigma_\text{w}$ and $\Sigma_\text{p}$ if the noise bound from Assumption~\ref{Noise1} is tight.

\begin{thm}[Asymptotic accuracy of $\mathit\Sigma_\text{w}$]\label{ConverW}
Under Assumption~{\ref{Noise1} with $\tilde{d}_i^T\tilde{d}_i\leq\epsilon^2, \epsilon>0$ and Assumption}~\ref{AsstightBound}, \ref{AssIndepend}, and \ref{AssPE}, the superset $\Sigma_\text{w}$ is a subset of
\begin{equation*}
	\{F^*\}\oplus\epsilon(L-1)\sqrt{\frac{L_\text{pe}}{\beta}}\mathcal{B}
\end{equation*}
with probability one for $S\rightarrow\infty$.
\end{thm}
\begin{proof}
The statement follows by Lemma~\ref{SetofunFaCoe} for $L_0=L$. It remains to compute $\rho$ of Assumption~\ref{AsstightBound} if we suppose the average bound \eqref{DataWindow} instead of the tight pointwise description from Assumption \ref{Noise1} which is equivalent to \eqref{DualNoise} with $\Delta_{1,i}=-1/\epsilon^2$, $\Delta_{2,i}=0$, and $\Delta_{3,i}=I_{n_x}$ for $i=1,\dots,S$. Note that $\rho$ can be specified by
\begin{align*}
	\rho^2&=\max_{\hat{D}\in\mathcal{D}_\text{nt}}\min_{\bar{D}\in\Omega}||\hat{D}-\bar{D}||_\text{Fr}^2\\
	&=\max_{\hat{D}\ \text{satisfies}\,\eqref{DataWindow}}\min_{\substack{\bar{d}_i\ \text{satisfies}\ \eqref{DualNoise}\\i=1,\dots,L}}\sum_{i=1}^{L}||\hat{d}_i-\bar{d}_i||^2_2
\end{align*}
with $\hat{D}=\begin{bmatrix}\hat{d}_1 & \cdots & \hat{d}_L\end{bmatrix}$ and $\bar{D}=\begin{bmatrix}\bar{d}_1 & \cdots & \bar{d}_L\end{bmatrix}$. For the considered special case of $\Delta_{1,i}$, $\Delta_{2,i}$, and $\Delta_{3,i}$, the minimizing $\bar{d}_i$ are given by $\epsilon\hat{d}_i/||\hat{d}_i||_2$ as $\bar{d}_i$ lie within balls with radius $\epsilon$. To solve the remaining maximization, observe that the point $x^*=\begin{bmatrix}0 & \cdots &L & \cdots & 0 \end{bmatrix}^T$ maximizes within $\{x\in\mathbb{R}^L:||x||_2\leq L\}$ the distance to any point in the $\infty$-norm unit ball $\{x\in\mathbb{R}^L:||x||_\infty\leq1\}$. Thereby, the energy of the maximizing realization of $\hat{D}$ is concentrated into one time point, i.e., $\hat{D}=\begin{bmatrix}0 & \cdots & \hat{d}_k & \cdots & 0\end{bmatrix}$. This yields
\begin{align*}
	\rho^2&{=}\max_{\begin{bmatrix}0 & \cdots & \hat{d}_k & \cdots & 0\end{bmatrix}\ \text{satisfies}\,\eqref{DataWindow}}\left(1-\frac{\epsilon}{||\hat{d}_k||_2}\right)^2||\hat{d}_k||^2_2\\
	&{=}\left(1-\frac{\epsilon}{\epsilon L}\right)^2(\epsilon L)^2=\epsilon^2(L-1)^2.
\end{align*}
\end{proof}

For the frequently-assumed case of noise with bounded amplitude, Theorem~\ref{ConverW} shows that $\rho$ is zero for window length one and increases with increasing window length. 
Thus, the accuracy of the window-based description \eqref{DataWindow} decreases for larger window length as $\rho$ measures the tightness of the supposed noise description. On the other hand, the number of windows decreases for larger window lengths which achieve less required optimization variables in the determination of input-output properties. Furthermore, Lemma~\ref{SetofunFaCoe} clarifies that $\Sigma_\text{w}$ converges to $\{F^*\}$ if the window noise bound \eqref{DataWindow} is tight.\\\indent 
The following theorem shows that the diameter of superset $\Sigma_\text{p}$ not only decreases monotonically with $S$ but actually converges to zero.

\begin{thm}[Asymptotic consistency of $\mathit\Sigma_\text{p}$]\label{Converp}
Under Assumption~\ref{Noise1}, \ref{AsstightBound}, \ref{AssIndepend}, and \ref{AssPE}, the superset $\Sigma_\text{p}$, from Proposition~\ref{Primal_elli} with maximal $\kappa>0$ for $\Delta_{1\text{p}}\succeq\kappa I_{n_z}$, converges to $\{F^*\}$ with probability one for data \eqref{DataSet} with  $S\rightarrow\infty$.	
\end{thm}
\begin{proof}
Since $\Sigma_\text{p}$ is calculated based on the tight noise characterization from Lemma~\ref{DualSigma}, $\rho=0$. Hence, Lemma~\ref{SetofunFaCoe} for $L_0=1$ and $\rho=0$ implies that $\bigcap_{t=1}^{S}F_t\rightarrow \{F^*\}$ for $S\rightarrow\infty$ with probability $1$, and thus the diameter of $\bigcap_{t=1}^{S}F_t$ tends to zero. Thus, the sequence of $\Sigma_\text{p}\supseteq\bigcap_{t=1}^{S}F_t$ with minimal diameter tends to  
	$\{F^*\}$ for $S\rightarrow\infty$.
\end{proof}



To prove consistency of $\Sigma_\text{c}$ for $S\rightarrow\infty$ with a tight noise description, Lemma~\ref{SetofunFaCoe} is not feasible, and therefore we adapt the results from \cite{SetMem_Con} for an unknown coefficient matrix and a more general quadratic noise characterization. For that purpose, we define the matrices
\begin{align*}
	X^+_\text{c}(S)&=\begin{bmatrix}\tilde{x}_1^+ &\cdots & \tilde{x}_{S}^+\end{bmatrix},\\
	Z_\text{c}(S)&=\begin{bmatrix}z(\tilde{x}_1,\tilde{u}_1) &\cdots & z(\tilde{x}_{S},\tilde{u}_{S})\end{bmatrix},\\
	D_\text{c}(S)&=\begin{bmatrix}\tilde{d}_1 &\cdots & \tilde{d}_{S}\end{bmatrix}
\end{align*} 
with $X^+_\text{c}(S)=F^{*}Z_\text{c}(S)+D_\text{c}(S)$ and assume that the noise realizations $D_\text{c}(S)$ are an element of the compact set 
\begin{equation}\label{NontightAvBound}
	\mathcal{D}_\text{nt,c}(S)=\left\{D\in\mathbb{R}^{n_x\times S}{:}\begin{bmatrix}D^T\\I_{n_x}\end{bmatrix}^T{\Delta}_{\text{nt}}(S)\begin{bmatrix}D^T\\I_{n_x}\end{bmatrix}\preceq0 \right\}
\end{equation}
with $\Delta_{\text{nt}}(S)=\text{diag}\left({\Delta}_{1,\text{t}}(S), {\Delta}_{3,\text{t}}(S)+{\Delta}_{3,\text{nt}}(S)\right)$ and $\Delta_{3,\text{nt}}(S)\preceq0$. While $\mathcal{D}_\text{nt,c}(S)$ corresponds to a known but not tight noise description, we assume that there exists an unknown tight cumulative bound on the noise realizations. 

\begin{assu}[Tight cumulative noise bound]\label{AsstightAvBound}
Suppose that the noise bound 
\begin{equation*}
	\mathcal{D}_\text{t,c}(S)=\left\{D\in\mathbb{R}^{n_x\times S}{:}\begin{bmatrix}D^T\\I_{n_x}\end{bmatrix}^T{\Delta}_{\text{t}}(S)\begin{bmatrix}D^T\\I_{n_x}\end{bmatrix}\preceq0 \right\},
\end{equation*}	
with ${\Delta}_{\text{t}}(S)=\text{diag}\left({\Delta}_{1,\text{t}}(S), {\Delta}_{3,\text{t}}(S)\right)$ and $\Delta_{1,\text{t}}(S)\succ0\ \forall S\in\mathbb{N}$, is a tight bound of $D_\text{c}(S)$ for $S\rightarrow\infty$, i.e., there exists a sequence $\{S_k\}_{k\in\mathbb{N}}$ of integers with $S_k\rightarrow\infty$ for $k\rightarrow\infty$ such that for any $\rho>0$
\begin{equation}\label{ConProb}
	\lim_{k\rightarrow\infty}\text{Pr}\left(\begin{bmatrix}D_\text{c}(S_k)^T\\I_{n_x}\end{bmatrix}^T{\Delta}_{\text{t}}(S_k)\begin{bmatrix}D_\text{c}(S_k)^T\\I_{n_x}\end{bmatrix}\succeq -\rho I_{n_x}\right)=1.
\end{equation}
\end{assu}

Assumption~\ref{AsstightAvBound} requires a tight cumulative noise bound $\mathcal{D}_\text{t,c}(S)$ for $S\rightarrow\infty$ which however does not exist for the pointwise bounded noise in Assumption~\ref{Noise1}. Nonetheless, we show asymptotic consistency of $\Sigma_\text{c}$ as cumulative noise bounds are commonly supposed, exemplary, in \cite{AnneDissi} and \cite{Groningen}.

\begin{thm}[Asymptotic accuracy of $\mathit\Sigma_\text{c}$]\label{SetofunFaCoeAv}
Under Assumption~\ref{AssPE}, \ref{AsstightAvBound}, and
\begin{equation}\label{Assum3}
	\text{Pr}\left(\lim_{S\rightarrow\infty}\frac{1}{S}\bigg|\bigg|Z_\text{c}(S){\Delta}_{1,\text{t}}(S)D_\text{c}(S)^T\bigg|\bigg|_\text{Fr}=0 \right)=1,
\end{equation}
the coefficients $F\in\mathbb{R}^{n_x\times n_z}$ feasible with the data \eqref{DataSet} and the non-tight noise bound~\eqref{NontightAvBound} satisfy
\begin{equation*}
	\lim_{j\rightarrow\infty} ||F^*-F||_\text{Fr}^2+\frac{1}{\gamma\beta\left\lfloor\frac{S_{k_j}}{L_\text{pe}}\right\rfloor}\text{tr}\left({\Delta}_{3,\text{nt}}(S_{k_j})\right)\leq0
\end{equation*}
with probability one for some $\gamma>0$ with ${\Delta}_{1,\text{t}}(S)\succeq\gamma I_{S}\ \forall S\in\mathbb{N}$ and some subsequence $\{S_{k_j}\}_{j\in\mathbb{N}}$ of sequence $\{S_k\}_{k\in\mathbb{N}}$ from Assumption~\ref{AsstightAvBound}.
\end{thm}  
\begin{proof}
While the overall idea follows from \cite{SetMem_Con}, we provide additional references and consider the case of a more general noise description using a matrix notation. For the sake of short notation, we omit to some extent the dependence on $S$. \\
From the given data \eqref{DataSet} and the non-tight noise bound~\eqref{NontightAvBound}, we conclude that the coefficients $F\in\mathbb{R}^{n_x\times n_z}$ admissible with the data are given by
\begin{align*}
	\Sigma_\text{c}(S)=\left\{F:\star^T{\Delta}_{\text{nt}}(S)\begin{bmatrix}(X^+_\text{c}(S)-FZ_\text{c}(S))^T\\I_{n_x}\end{bmatrix}\preceq0 \right\}.
\end{align*}
Together with $X^+_\text{c}(S)=F^*Z_\text{c}(S)+D_\text{c}(S)$, any coefficients $F\in\Sigma_{\text{c}}(S)$ satisfy
\begin{align}\label{Inequ3}
	&\begin{bmatrix}D_\text{c}^T\\I_{n_x}\end{bmatrix}^T{\Delta}_{\text{t}}\begin{bmatrix}D_\text{c}^T\\I_{n_x}\end{bmatrix}\notag\\
	\preceq&-\tilde{F}Z_\text{c}{\Delta}_{1,\text{t}}Z_\text{c}^T\tilde{F}^T-D_\text{c}{\Delta}_{1,\text{t}}Z_\text{c}^T\tilde{F}^T-\tilde{F}Z_\text{c}{\Delta}_{1,\text{t}}D_\text{c}^T-{\Delta}_{3,\text{nt}}\notag\\
	\preceq&-\gamma\tilde{\beta}\, \tilde{F}\tilde{F}^T-D_\text{c}{\Delta}_{1,\text{t}}Z_\text{c}^T\tilde{F}^T{-}\tilde{F}Z_\text{c}{\Delta}_{1,\text{t}}D_\text{c}^T-{\Delta}_{3,\text{nt}}
\end{align}
with $\tilde{F}=F^*-F$ and $\tilde{\beta}=\left\lfloor\frac{S_k}{L_\text{pe}}\right\rfloor\beta$. The second inequality holds because ${\Delta}_{1,\text{t}}(S)\succeq\gamma I_{S}\ \forall S\in\mathbb{N}$ and Assumption~\ref{AssPE}, which implies  
\begin{equation*}
	Z_\text{c}Z_\text{c}^T=\sum_{i=1}^{S_k-1} z(\tilde{x}_i,\tilde{u}_i)z(\tilde{x}_i,\tilde{u}_i)^T  \succeq\tilde{\beta} I_{n_z}.
\end{equation*}
\eqref{Inequ3} together with \eqref{ConProb} and the fact that the convergence with probability one in \eqref{Assum3} implies convergence in probability \cite{Stochastic} (Theorem 17.2), yields for the sequence $\{S_k\}_{k\in\mathbb{N}}$ from Assumption~\ref{AsstightAvBound} that
\begin{equation*}
	\lim_{k\rightarrow\infty}\text{Pr}\left(||F^*-F||_\text{Fr}^2\leq \frac{\rho n_x}{\gamma\tilde{\beta}} -\frac{1}{\gamma\tilde{\beta}}\text{tr}\left({\Delta}_{3,\text{nt}}(S_k)\right)\right)=1
\end{equation*}
for any $\rho>0$ and any coefficients $F$ feasible with the data, i.e., $F\in\Sigma_\text{c}(S_k)$. Finally, according to \cite{Stochastic} (Theorem 17.3), there exists a subsequence $\{S_{k_j}\}_{j\in\mathbb{N}}$ of $\{S_k\}_{k\in\mathbb{N}}$ with
\begin{equation*}
	\text{Pr}\left(\lim_{j\rightarrow\infty} ||F^*-F||_\text{Fr}^2+\frac{1}{\gamma\tilde{\beta}}\text{tr}\left({\Delta}_{3,\text{nt}}(S_{k_j})\right)\leq0\right)=1.
\end{equation*} 
\end{proof}

First, note that the additional assumption \eqref{Assum3} corresponds to the average noise property in \cite{SetMem_Con} (Theorem 2.3) and is satisfied exemplary for zero mean noise.\\\indent
Second, if the cumulative noise description \eqref{NontightAvBound} is actually tight, i.e., ${\Delta}_{3,\text{nt}}(S)=0$ for $S\rightarrow\infty$, then Theorem~\ref{SetofunFaCoeAv} shows that $||F^*-F||_F$ converges to zero with probability one. For that reason, $\Sigma_{\text{c}}$ is asymptotically consistent. Furthermore, if a zero mean noise signal has a covariance matrix $\delta I_{n_x}$ but is overestimated by $\tilde{\delta}I_{n_x}$, i.e., $\lim_{t\rightarrow\infty}\frac{1}{t}\sum_{i=1}^{t}\tilde{d}_i\tilde{d}_i^T=\delta I_{n_x}\leq(\delta+\tilde{\delta})I_{n_x}$, then ${\Delta}_{1,\text{t}}(S)=I_{S}$ and ${\Delta}_{3,\text{nt}}(S)=-S\tilde{\delta}I_{n_x}$, and hence Theorem~\ref{SetofunFaCoeAv} implies that any coefficients $F$ feasible for infinitely many samples are contained in $\left\{F:||F^*-F||_\text{Fr}^2\leq\frac{L_\text{pe}\tilde{\delta}n_x}{\beta}\right\}$ with probability one.

\subsection{Comparison of the accuracy of $\mathit\Sigma_\text{p}, \mathit\Sigma_\text{w}$, and $\mathit\Sigma_\text{c}$ in a numerical example}\label{NumComparison}

To assess the accuracy of the three supersets $\Sigma_{\text{p}}$, $\Sigma_{\text{w}},$ and $\Sigma_{\text{c}}$ for pointwise bounded noise \eqref{SepPrimnoise}, we consider a data-driven estimation of the $\ell_2$-gain of a polynomial system. To this end, we apply \cite{MartinDissi} (Theorem~2) for the three supersets and compare the results with the $\ell_2$-gain derived directly from the system dynamics by SOS optimization and the $\ell_2$-gain calculated from \cite{MartinDissi} (Corollary~1), where pointwise bounded noise can be exploited directly in the data-driven computation of the $\ell_2$-gain. In particular, we evaluate the $\ell_2$-gain of
\begin{align}\label{NumEx}
	\begin{bmatrix}x_1(t+1)\\x_2(t+1)
	\end{bmatrix}&=\begin{bmatrix}
	-0.3x_1+0.2x_2^2+0.2x_1x_2\\
	0.2x_2+0.1x_2^2-0.3x_1^3+0.4u
	\end{bmatrix}(t)
\end{align} 
for $u\rightarrow x$ within the operation set $x_1^2\leq 1$, $x_2^2\leq 1$, and $u^2\leq \sqrt{2}$. We draw samples \eqref{DataSet} from a single trajectory with initial condition $x(0)=\begin{bmatrix}-1 & -1\end{bmatrix}^T$, $u(t)=1.5\sin(0.002t^2+0.1t)$, and noise that exhibits constant signal-to-noise-ratio $||\tilde{d}_i||_2\leq0.02||\tilde{x}_i||_2$. Moreover, we assume $z(x,u)=\begin{bmatrix}x_1&x_2& x_2^2&x_1x_2&x_1^3&u\end{bmatrix}^T$. Table~\ref{ErrorBound} shows the received upper bounds on the $\ell_2$-gain.

\begin{table}[hb]
	\begin{center}
		\caption{Data-driven inference on the $\ell_2$-gain of \eqref{NumEx}.}\label{ErrorBound}
		\begin{tabular}{c|ccc}\label{TableSim}
			& $S=20$ & $S=50$ & $S=100$\\\hline
			model-based & $0.5814$ & $0.5814$ & $0.5814$ \\
			\cite{MartinDissi} (Corollary~1)  & $0.8271$  & $0.5997$ & $0.5983$ \\
			$\Sigma_{\text{p}}$ (minimal diameter) & $2.1069$ & $0.7251$ & $0.7004$ \\
			$\Sigma_{\text{w}}$ $(L=10)$ & $9.3376$ &  $1.0156$ & $0.7917$ \\
			$\Sigma_{\text{w}}$ $(L=20)$ & $\infty$ & $1.0589$ & $0.9119$  \\
			$\Sigma_{\text{c}}$ & $\infty$ & $2.2894$ & $3.8952$  
		\end{tabular}
	\end{center}
\end{table}

As expected, the upper bounds from \cite{MartinDissi} (Corollary~1) differ by the smallest margin from the model-based upper bound. However, the computation times with $52\,\text{s}$ ($S=20$), $61\,\text{s}$ ($S=50$), and $104\,\text{s}$ ($S=100$) are more demanding than the computation times of less than a second for \cite{MartinDissi} (Theorem~2) with $\Sigma_{\text{p}}, \Sigma_{\text{w}},$ and $\Sigma_{\text{c}}$. Table~\ref{TableSim} also shows that $\Sigma_{\text{p}}$ outperforms the other supersets and that the accuracy of $\Sigma_{\text{w}}$ increases with decreasing window length $L$, as expected by Theorem~\ref{ConverW}. Note that the increase of the upper bounds by $\Sigma_{\text{c}}$ for increasing $S$ is already excessively discussed in \cite{MartinDissi}.\\\indent
Summarized, we prefer $\Sigma_{\text{p}}$ over $\Sigma_{\text{w}}$ and $\Sigma_{\text{c}}$ to obtain data-driven inference on input-output properties if the noise exhibits pointwise bounds and the number of samples allows to solve LMI~\eqref{Elli_LMI}.

\section{Conclusions}

By Algorithm~\ref{Algorithm1}, we established a set-membership framework to determine optimal input-output properties of polynomial systems without identifying an explicit model but directly from input-state measurements in the presence of noise. In particular, we focused on guaranteed upper bounds on NLMs of dynamical systems and their \textquoteleft optimal\textquoteright\ linear approximation as well as on input-output properties specified by time domain hard IQCs. We emphasize that the framework achieves computationally tractable LMI conditions with SOS multipliers even regarding to the unknown linear filter. Related to the set-membership literature, we also presented three data-driven supersets that include the true unknown coefficient matrix and showed their asymptotic consistency. \\\indent
While the framework is presented for polynomial systems, it can be extended to nonlinear systems by \cite{MartinNonlinear}. Indeed, the polynomial sector bounds from \cite{MartinNonlinear} include the unknown nonlinear system dynamics, are derived from data without knowledge of the true basis functions, and are suitable for the here applied robust control techniques. Subject of future research is the application of the proposed framework in practice and a thorough comparison of deterministic and stochastic approaches for determining system properties which could result in a framework fusing the advantages of both.

\vspace{-1.5cm}

\begin{IEEEbiography}[{\includegraphics[width=1.3in,height=1.25in,clip,keepaspectratio]{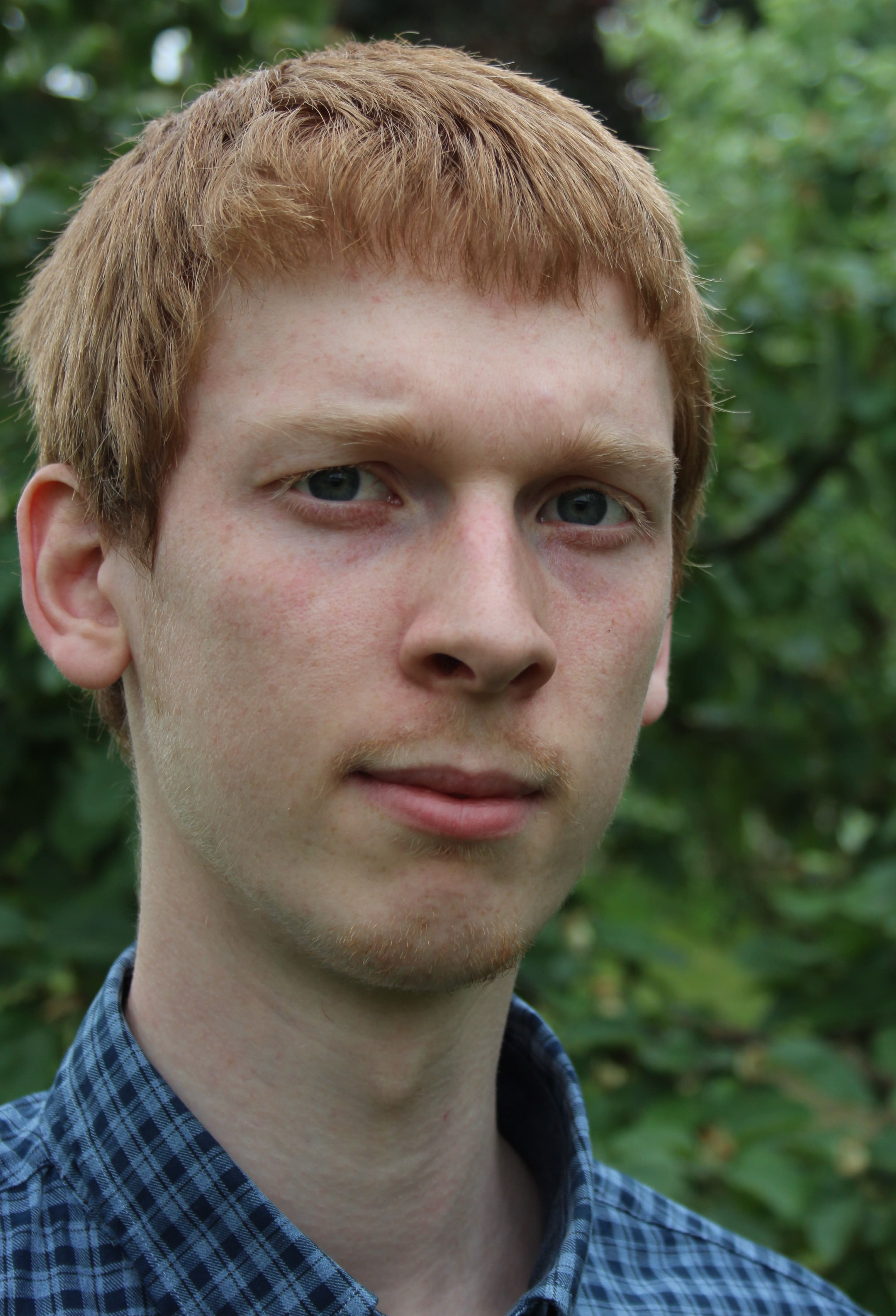}}] 
{Tim Martin} received the Master’s degree in Engineering Cybernetics from the University of Stuttgart, Germany, in 2018. Since 2018, he has been a Research and Teaching Assistant at the Institute for Systems Theory and Automatic Control and a member of the Graduate School Simulation Technology at the University of Stuttgart. His research interests include data-driven system analysis and control with focus on nonlinear systems.
\end{IEEEbiography}
\vspace{-15cm}
\begin{IEEEbiography}[{\includegraphics[width=1.2in,height=1.25in,clip,keepaspectratio]{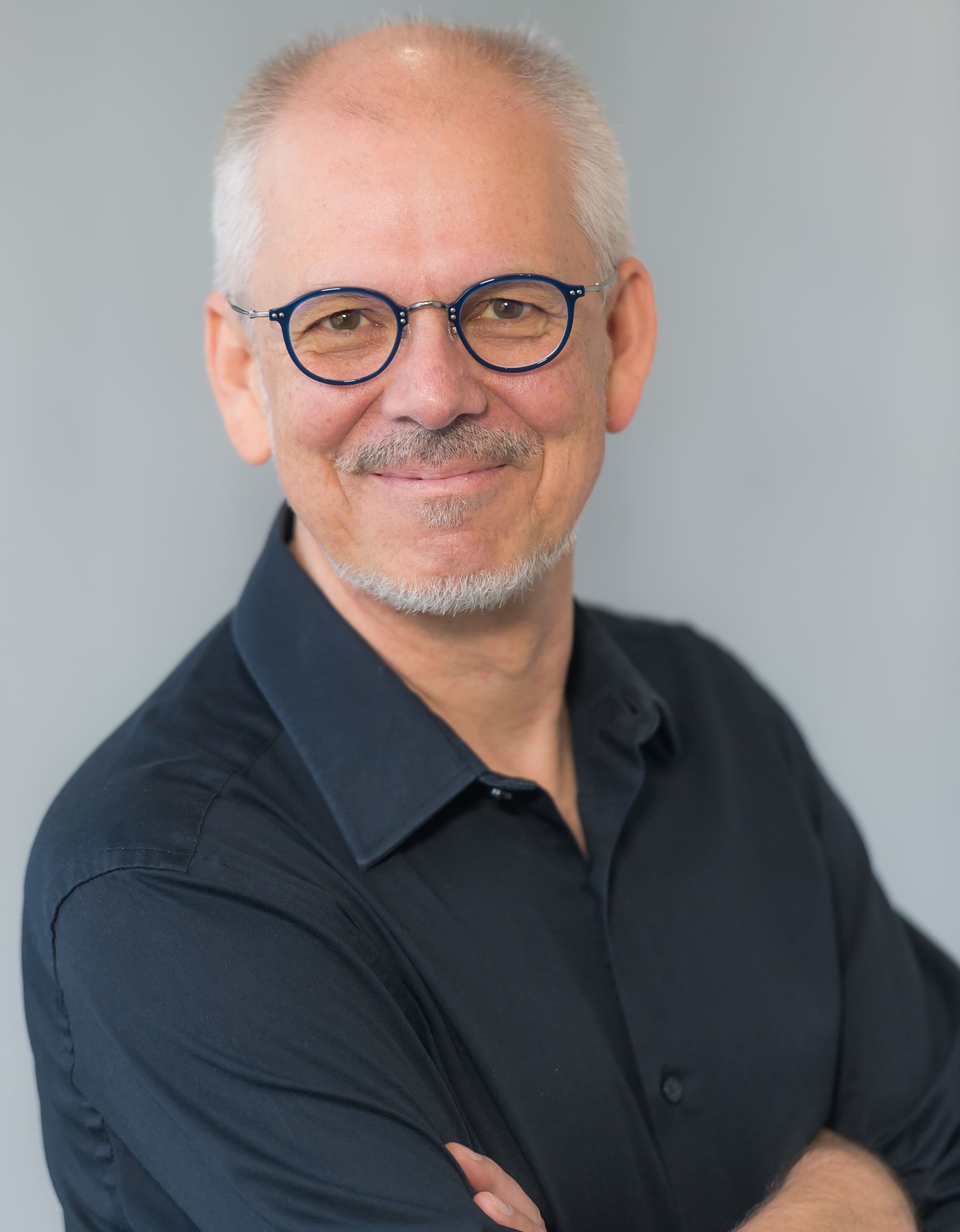}}]
{Frank Allgöwer} is professor of mechanical engineering at the University of Stuttgart, Germany, and Director of the Institute for Systems Theory and Automatic Control (IST) there.\\\indent
He is active in serving the community in several roles: Among others he has been President of the International Federation of Automatic Control (IFAC) for the years 2017-2020, Vicepresident for Technical Activities of the IEEE Control Systems Society for 2013/14, and Editor of the journal Automatica from 2001 until 2015. From 2012 until 2020 he served in addition as Vice-president for the German Research Foundation (DFG), which is Germany’s most important research funding organization.\\\indent
His research interests include predictive control, data-based control, networked control, cooperative control, and nonlinear control with application to a wide range of fields including systems biology.	
\end{IEEEbiography}

\end{document}